\documentclass[11pt]{article}
\usepackage{geometry}                % See geometry.pdf to learn the layout options. There are lots.
\usepackage{graphicx}

\usepackage[utf8]{inputenc} % Any characters can be typed directly from the keyboard, eg éçñ
\usepackage{textcomp} % provide lots of new symbols

\usepackage{flafter}  % Don't place floats before their definition

\usepackage{amsmath,amssymb}  % Better maths support & more symbols
\usepackage{bm}  % Define \bm{} to use bold math fonts
\usepackage{amsthm}

\usepackage{memhfixc}  % remove conflict between the memoir class & hyperref

\newtheorem{definition}{Definition}[section]
\newtheorem{theorem}{Theorem}[section]
\newtheorem{proposition}{Proposition}[section]
\newtheorem{corollary}{Corollary}[section]
\newtheorem{remark}{Remark}[section]
\newtheorem{lemma}{Lemma}[section]
\numberwithin{equation}{section}

\title{Hopf Algebras Concerning Matrices or Finite Sets
        and Their Application to Star Product of Scalar fields}
\author{Zhou Mai \footnote{address:Colleague of Mathematical Science, Nankai University, Weijin Road, Tianjin City, Republic China;
            email address: zhoumai@nankai.edu.cn}}

\begin{document}

\maketitle

\begin{abstract}
In this article the Hopf algebra structure concerning the finite 
sets is presented. Here the crucial stage is the operation
of quotient about finite subsets, the sequences of disjoint subsets  
and power sets of some finite set following the ideas of
Connes-Kreimer Hopf algebra (\cite{1,2}), hence,  
the construction in the present article is the generalization of Connes-Kreimer
Hpof algebra consisting of Feynmam diagrams.  
As applications based on the construction concerning
abstract finite sets, the Hopf algebras consisted of matrices,
or, scalar fields under the star product is constructed
as well.
\end{abstract}

\tableofcontents

\section{Introduction}

In the present article we generalise Connes-Kreimer
Hpof algebra (see A. Connes an D. Kreimer \cite{1,2})
consisting of Feynmam diagrams to the situations of abstract
finite sets, matrices and star product of scalar field,
where the construction for the case of finite sets is essential.
The heartening observation in 
H. Figueroa and J.M. Gracia-Bondia\cite{3}
(see J. M. Gracia-Bondia, J. C. Varilly and
H. Figueroa \cite{4} also)
shows that Connes-Kreimer's coproduct of
Feynman diagrams can be admitted to subgraphs
such that the co-associativity is available still.
A subgraph of a connected Feynman diagram is a
subdiagram determined by its vertices completely.
Somehow we can centre on vertices for the structure
of Hopf algebra of Feynman diagrams.
This is our motivation to generalise the structure of Hopf
algebra about Feynman diagrams to more general cases.
Roughly speaking, a abstract finite
set can be viewed as a set of "vertices".

In our setting the key issue is the notation of quotient
which is the generalisation of similar notation of 
Feynman diagrams.
In order to generalise the structures concerning Feynman
diagrams we need to "translate" some other notations of Feynman
diagrams into language of set theory, or, matrices and star product. 
For example, subgraphs are translated as subsets or sequences
of disjoint union subsets, factorisations of Feynman amplitudes
are translated to be partitions of some subset. 
We establish those notations
in a formal way from pure algebraic viewpoint,
actually, under our consideration there is not sub-divergence
to be considered. In our setting the
crucial part is to construct the coproduct which is
co-associative and co-nilpotent such that the
tenser algebra and symmetric tenser algebra of
the coalgebra become Hopf algebra.

The article is organized as following. In the section 2
the notations of quotient or collapsing for 
subsets are discussed in details.
In the section 3 we construct two types of coproducts
in finite sets. 
In the section 4 we discuss the case of matrices.
Finally, in the section 5 we discuss star product.

\section{Quotient and collapsing of finite sets}

In this article every issue can be reduced to the case of
finite sets, therefore, we restrict our consideration in the 
case of finite sets only.

\subsection{Partitions}

For an abstract finite set $A$, $\#A=d$($d$ is a positive integer), let
$(I_{1},\cdots, I_{l})$ be a sequence of disjoint non-trivial
subsets in $A$,
it is also denoted by $(I_{i})$ for short, we can always regard $(I_{i})$ as a
partition of some set. Actually, let $I=\bigcup_{i=1}^{l}I_{i}$, then
$(I_{i})$ is a partition of $I$, i.e. $(I_{i})\in \mathbf{Part}(I)$,
where $\mathbf{Part}(B)$ denotes the set of all partitions of
some finite set $B$, i.e.

$$
\mathbf{Part}(B)=\{\{I_{i}\}|\,\bigcup_{i}I_{i}=B,
I_{i}\cap I_{i^{\prime}}=\emptyset,i\not= i^{\prime}\}.
$$
Thus, for simplicity, we call the sequence of disjoint non-trivial subsets
of $A$ the partition in $A$ below. For a partition $(I_{i})$ and a
subset $U\subset A$ we say $(I_{i})$ is in $U$, denoted by 
$(I_{i})\subset U$, if $\bigcup_{i}I_{i}\subset U$.

For two partitions $(I_{1},\cdots, I_{l})$
and $(J_{1},\cdots, J_{k})$ in $A$, let $I=\bigcup_{i=1}^{l}I_{i}$
and $J=\bigcup_{j=1}^{k}J_{j}$, we need the following notations:

\begin{itemize}
  \item \textbf{Joint:} $(I_{i}\cap J_{j})$ is a partition of $I\cap J$
  denoted by $(I_{i})\cap(J_{j})$. If $I\cap J=\emptyset$ we say $(I_{i})$ and
  $(J_{j})$ are disjoint denoted by $(I_{i})\cap(J_{j})=\emptyset$.
  \item \textbf{Union:} If $(I_{i})\cap(J_{j})=\emptyset$, 
  $(I_{1},\cdots,I_{l};J_{1},\cdots,J_{k})$ is a partiton of $I\cup J$ denoted
  by $(I_{i})\cup(J_{j})$.
  \item \textbf{Inclusion:} If for $\forall\,I_{i},\exists\,J_{j}$, such that
  $I_{i}\subset J_{j}$, we say $(J_{j})$ includes $(I_{i})$ denoted by
  $(I_{i})\subset(J_{j})$. In this case we call $(I_{i})$ is a sub-partition
  of $(J_{j})$.
\end{itemize}

Moreover, we introduce some notations as follows:

\begin{itemize}
  \item The issues in this article involve the power-set closely,
  let $\mathcal{P}(A)$ denote the power-set of $A$, we set
  $$
  \mathcal{P}^{k+1}(A)=\mathcal{P}(\mathcal{P}^{k}(A)),
  \mathcal{P}^{1}(A)=\mathcal{P}(A),
  $$
  where $k$ is a positive integer. Then we have
  $$
  B\in\mathcal{P}^{k+1}(A)\,\Longleftrightarrow\,
  B\subset\mathcal{P}^{k}(A).
  $$
  If $B\in\mathcal{P}^{k}(A)$ we say $B$ is provided
  with power degree $k$.
  \item Let $\mathcal{P}^{2}_{dis}(A)$ denote the set
  of partitions in $A$,
  $$
  \mathcal{P}^{2}_{dis}(A)=\{\{I_{1},\cdots,I_{l}\}\in\mathcal{P}^{2}(A)
  |I_{i}\in A,\,I_{i}\cap I_{j}=\emptyset,\,i\neq j\}.
  $$
  We assume $\{\emptyset\}\in\mathcal{P}^{2}_{dis}(A)$.
  
  We define a map from $\mathcal{P}^{2}_{dis}(A)$ to $\mathcal{P}(A)$
  as following:
  
  \begin{equation}
  \mathcal{R}:\mathcal{P}^{2}_{dis}(A)\longrightarrow\mathcal{P}(A),\,
  \mathcal{R}:\{I_{1},\cdots,I_{l}\}\mapsto \bigcup\limits_{i=1}^{l}I_{i},
  \end{equation}
  where $\{I_{i}\}\in\mathcal{P}^{2}_{dis}(A)$. We call $\mathcal{R}$
  the \textbf{reversion map} which decreases the power degree of
  a set. 
  \item Let $\mathcal{P}(A)\times_{dis}\mathcal{P}^{2}_{dis}(A)$
  denote a subset of $\mathcal{P}(A)\times\mathcal{P}^{2}_{dis}(A)$,
  for $\{U\}\in\mathcal{P}(A)$ and 
  $\{I_{1},\cdots,I_{l}\}\in\mathcal{P}^{2}_{dis}(A)$,
  
  $$
  (\{U\},\{I_{i}\})\in\mathcal{P}(A)\times_{dis}
  \mathcal{P}^{2}_{dis}(A)\Longleftrightarrow
  U\cap\mathcal{R}(\{I_{i}\})=\emptyset.
  $$
  $\mathcal{P}(A)\times_{dis}\mathcal{P}^{2}_{dis}(A)$
  is also denoted by $\Xi_{A}$ for short.
  $\Xi_{A}$ will play the important role in this
  article.
  Usually we denote $(\{U\},\{I_{i}\})$ by $U\cup\{I_{i}\}$
  without confusion, actually, $\{U\}$ and $\{I_{i}\}$
  are provided with different power degree, thus
  here the disjoint union occurs only.
  Moreover, the reversion map $\mathcal{R}$ can be
  extended to a map from $\Xi_{A}$ to
  $\mathcal{P}(A)$,

\begin{equation}
\mathcal{R}_{1}(U\cup\{I_{i}\})=U\cup\mathcal{R}(\{I_{i}\}).
\end{equation}

  \item We define $\mathcal{P}_{dis}(\Xi_{A})
  \subset\mathcal{P}(\Xi_{A})$ as following:
  
  $$
  \begin{array}{c}
  \{I_{1}\cup J_{1},\cdots,I_{l}\cup J_{l}\}\in \mathcal{P}_{dis}(\Xi_{A})
  \,\Longleftrightarrow \\
  J_{i}\in\mathcal{P}_{dis}^{2}(A),\,i=1,\cdots,l,
  \{I_{i}\},\,\{\mathcal{R}(J_{i})\}\in\mathcal{P}_{dis}^{2}(A),\,
  \mathcal{R}(\{I_{i}\})\cap\mathcal{R}(\{\mathcal{R}(J_{i})\})
  =\emptyset.
  \end{array}
  $$
  We can extend the reversion map to $\mathcal{P}_{dis}(\Xi_{A})$,
  denoted by $\mathcal{R}_{1}$ also,
  \begin{equation}
  \mathcal{R}_{1}(\{I_{i}\cup J_{i}\})
  =\{\mathcal{R}_{1}(I_{i}\cup J_{i})\}
  =\{I_{i}\cup\mathcal{R}(J_{i})\}.
  \end{equation}
  $\mathcal{R}_{1}$ is map from $\mathcal{P}_{dis}(\Xi_{A})$ to
  $\mathcal{P}_{dis}^{2}(A)$ which decreases the power degree
  of sets also.
  
  \item Let $U\cup\{K_{\lambda}\}_{1\leq\lambda\leq l}\in\Xi_{A}$,
  $\{I_{i}\cup J_{i}\}_{1\leq i\leq k}\in\mathcal{P}_{dis}(\Xi_{A})$,
  we say $\{I_{i}\cup J_{i}\}\subset U\cup\{K_{\lambda}\}$ if
  $\mathcal{R}(\{I_{i}\})\subset U$, and 
  $\mathcal{R}(\{J_{i}\})\subset\{K_{\lambda}\}$, 
  where $\{K_{1},\cdots,K_{l}\}$ is regarded
  as a subset of $\mathcal{P}(A)$.
  
  \item Let $\{I_{i}\cup J_{i}\},\{K_{\lambda}\cup 
  L_{\lambda}\}\in\mathcal{P}_{dis}(\Xi_{A})$, we say
  $\{K_{\lambda}\cup L_{\lambda}\}\subset\{I_{i}\cup J_{i}\}$, 
  if for $\forall\lambda$, $\exists i$,
  such that $K_{\lambda}\cup L_{\lambda}\subset I_{i}\cup J_{i}$
  (i.e. $K_{\lambda}\subset I_{i}$ and $L_{\lambda}\subset J_{i}$).
\end{itemize}

\begin{remark}
$\\$
\begin{itemize}
\item In above discussions there are two ways to express
the partitions, which are $(I_{i})$ and $\{I_{i}\}$. In this article
the symbol $(\cdot)$ prefers the sequences of subsets,
and the symbol $\{\cdot\}$ prefer the subsets or elements
in power-set, for example, let $U$ be a subset of $A$, then
we have $U\subset A$ and $\{U\}\in \mathcal{P}(A)$.
\item Because $\mathcal{P}(A)$ and $\mathcal{P}^{2}(A)$
are completely different sets, we can identify 
$\mathcal{P}(A)\times\mathcal{P}^{2}(A)$ with
$\mathcal{P}(A)\cup\mathcal{P}^{2}(A)$ for convenience.
The more general cases are similar.
\item We can always allow that a partition consists of single set,
thus, in this sense, we are able to think of $\Xi_{A}$ as a subset
of $\mathcal{P}_{dis}(\Xi_{A})$.
\end{itemize}
\end{remark}

\subsection{Basic definition and properties of quotient}

For a given finite set $A$($\#[A]=d>0$),
let $U,I$ be subsets of $A$, we want to construct an operation
called quotient or collapsing which can be regarded
as a map:

$$
\begin{array}{ccc}
         & collapsing &    \\
\mathcal{P}(A)\times \mathcal{P}(A)& \longrightarrow  &  
\Xi_{A}=\mathcal{P}(A)\times_{dis} \mathcal{P}^{2}_{dis}(A). 
\end{array}
$$
We define the quotient
of $U$ by $I$ denoted by $U\diagup I$ in the following way.

\begin{definition}
Let $U,I\in A$.

\begin{itemize}
\item 

\begin{equation}
U\diagup I:=(U\setminus I)\cup \{U\cap I\}.
\end{equation}
\item Particularly,

$$
U\diagup I=
\left\{
\begin{array}{cc}
U\cup\{\emptyset\}, & I\cap U=\emptyset, \\
\emptyset\cup\{U\}, & U\subset I,
\end{array}
\right.
$$
where $\{\emptyset\}\in\mathcal{P}^{2}_{dis}(A)$.
\end{itemize}

In the above statements the symbol $\{\cdot\}$ denotes
an element in power-set $\mathcal{P}(A)$.
\end{definition}

\begin{remark}
$\\$
\begin{itemize}
  \item We call procedure from pair $(U,I)$ to quotient $U\diagup I$
  the collapsing. Roughly speaking, 
  the subset $I\cap U$ collapses to a new "ideal element".
  The notations of quotient and collapsing are motivated by the quotient and
  collapsing of Feynman diagrams.
  \item Particularly, we have
  $$
  U\diagup U=\emptyset\cup\{U\},\,U\diagup \emptyset=U\cup\{\emptyset\}.
  $$
  In the situation of $U\diagup\{\emptyset\}$, we identify
  $U\cup\{\emptyset\}$ with $U$.
  \item The definition 2.1 shows that the quotient $U\diagup I$
  is determined by $U\cap I$ solely. Actually,
  we can take $I^{\prime}=I\cap U$ instead of
  $I$, where we identify the $I$ with $I^{\prime}$ as  same "ideal element".
  The key idea is that the subset $I\cap U$ of $U$ collapses to a "ideal element".
  Generally, for two sets $I_{1}$, $I_{2}$,
  if $I_{1}\cap U=I_{2}\cap U$ we have $U\diagup I_{1}=U\diagup I_{2}$.
  In summary, the quotient of set concerns its subsets really.
  Without loss of generality, when we discuss the quotient 
  $U\diagup I$, we will assume $I\subset U$ in 
  discussion below.
  \item In definition 2.1 we express the quotient as $(U\setminus I)\cup \{I\cap U\}$,
  where the union involve the sets of different type, it is disjoint union always.
  Therefore, we can regard that as a pair $(U\setminus I,\{I\cap U\})$ belonging
  to $\mathcal{P}(A)\times\mathcal{P}^{2}(A)$.  
\end{itemize}
\end{remark}

For three finite sets $U,V,I\subset A$,
$I\subset U\cap V$,
it is obvious that we have
\begin{equation}
(U\diagup I)\cap(V\diagup I)=(U\cap V)\diagup I,
\end{equation}
and
\begin{equation}
(U\diagup I)\cup(V\diagup I)=(U\cup V)\diagup I.
\end{equation}
Particularly, if $I\subset U\subset V$ we have
$$
U\diagup I\subset V\diagup I.
$$

Now we consider the case of making quotient repeatedly.
Let $I\subset U\subset A$ and $J\subset U\diagup I$,  
we can make collapsing two times. By definition 2.1 we have

$$
(U\diagup I)\diagup J=(((U\setminus I)\cup \{I\})\setminus J)\cup \{J\}.
$$
Let $J^{\prime}=J\cap U$,
we discuss the different cases as follows:
\begin{itemize}
  \item $\{I\}\subset J$: $J=J^{\prime}\cup \{I\}$, noting
  $J^{\prime}\cap I=\emptyset$ due to $J\subset U\diagup I$
  we have
  $$
  (U\diagup I)\diagup J=((U\setminus I)\setminus J^{\prime})\cup \{J\}
  =((U\setminus (I\cup J^{\prime}))\cup \{J\},
  $$
  where $J\in \mathcal{P}(A)\times \mathcal{P}^{2}(A)$, i.e.
  there is a component with power degree 2 in $J$.
   \item $\{I\}\notin J$: In this case $J=J^{\prime}$
   and $J\cap I=\emptyset$, We have
   $$
   (U\diagup I)\diagup J=(U\setminus (I\cup J))\cup \{I,J\}.
   $$
   
   More generally, for tow subsets $I,\,J$, $I,J\in U$,
   $I\cap J=\emptyset$, we have
   
   $$
   (U\diagup I)\diagup J=(U\diagup J)\diagup I
   =(U\setminus (I\cup J))\cup\{I,J\}.
   $$
   Additionally, 
  if $U\cap V=\emptyset$, and $I\subset U$, $J\subset V$ we have
  
  \begin{equation}
  (U\diagup I)\cup(V\diagup J)=((U\cup V)\diagup I)\diagup J.
  \end{equation}
\end{itemize}

\subsection{Quotient by partitions}

Furthermore, we consider the situation of making
collapsing many times.

\paragraph{The quotient of one subset by partition}

Here we are interested in two special situations

\textbf{The first case:}
Let $I_{1},\cdots,$ $I_{l}\subset U$, here $U\subset A$,
in addition, we assume 
$$
I_{2}\subset U\diagup I_{1},\,
\cdots,\,I_{l}\subset (((U\diagup I_{1})\diagup I_{2})
\diagup\cdots)\diagup I_{l-1},
$$
above assumption means $I_{i}\cap I_{j}=\emptyset$
($i\neq j$).
Thus $(I_{i})$ is a partition in $U$.
We define
\begin{equation}
U\diagup (I_{i}):=(\cdots((U\diagup I_{1})\diagup I_{2}\cdots)\diagup I_{l}.
\end{equation}
By definition 2.1 we have
$$
U\diagup (I_{i})=(U\setminus I)\cup\{I_{i}\cap U|\,\#(I_{i}\cap U)>1,\,1\leqslant i\leqslant l\},
$$
where $I=\mathcal{R}(\{I_{i}\})$,
thus the quotient (2.8) dose not depend on the order of $I_{i}$.
There are formulas similar to (2.5), (2.6) and (2.7) in the case of partitions.

For general sequence of subsets 
$I_{1},\cdots,$ $I_{l}\subset U$ which satisfies
$I_{i}\setminus\bigcup_{j\neq i}I_{j}\neq\emptyset$
for esch $i$, we set $I_{i}^{\prime}=I_{i}\setminus
\bigcup_{j<i}I_{j}$ ($i>1$), $I_{1}^{\prime}=I_{1}$,
then $(I_{i}^{\prime})$ is a partition in $U$, and
it is easy to check that we have
$$
U\diagup (I_{i}^{\prime})
=(\cdots((U\diagup I_{1})\diagup I_{2}\cdots)\diagup I_{l}.
$$

\begin{remark}
$\\$
\begin{itemize}
\item Above discussion shows that when we discuss
the quotient by a sequence of subsets, we can always
reduce the situation to the one of partitions.
Actually, in this article we centre on the quotient by
partitions mainly.
\item
For two partitions $(I_{i})_{1\leqslant i\leqslant k}$
 and $(J_{i})_{1\leqslant i\leqslant k}$, if
$U\cap I_{i}=U\cap J_{i}$ ($1\leqslant i\leqslant k$),
we have $U\diagup (I_{i})=U\diagup (J_{i})$ due to similar reason in
remark 2.1.
\end{itemize}
\end{remark}

\textbf{The second case:}
We consider the following sequence of quotient
$$
U\diagup I_{1},\,(U\diagup I_{1})\diagup I_{2}
\cdots,\, (((U\diagup I_{1})\diagup I_{2})
\diagup\cdots)\diagup I_{l},
$$
but here the sequence of sets $\{I_{1},\cdots,I_{l}\}$
satisfies
$$
 \{I_{1}\}\subset I_{2},\cdots, \{I_{l-1}\}\subset I_{l}.
$$
Let
$$
I_{i}^{\prime}=I_{i}\cap((((U\diagup I_{1})\diagup I_{2})
\diagup\cdots)\diagup I_{i-1})
$$
By induction we can prove that
$$
\begin{array}{c}
I_{i}^{\prime}\subset U\setminus(I_{1}\cup I_{2}^{\prime}\cup
\cdots\cup I_{i-1}^{\prime}),  \\
(((U\diagup I_{1})\diagup I_{2})
\diagup\cdots)\diagup I_{i}=
(U\setminus(I_{1}\cup I_{2}^{\prime}\cup
\cdots\cup I_{i}^{\prime}))\cup\{I_{i}\}.
\end{array}
$$
Furthermore, we have
$$
I_{i}=I_{i}^{\prime}\cup\{I_{i-1}^{\prime}\cup\{\cdots
\cup\{I_{2}^{\prime}\cup\{I_{1}\}\}\cdots\}\},\,
i=2,\cdots,l.
$$

\paragraph{The quotient of the partition by partition}

We now turn to the discussion of the quotient of partitions by partitions.
Let $(I_{1},\cdots,I_{l})$ and $(J_{1},\cdots,J_{k})$ be two
partitions in $A$, we consider the quotient $(I_{i})\diagup (J_{j})$
which is defined as following 

$$
(I_{i})\diagup (J_{j})\doteq
(I_{i}\diagup (J_{j})_{J_{j}\subset I_{i}}).
$$
Due to remark 2.2 and 2.3, without loss of generality,
we can assume $(J_{j})\subset (I_{i})$
which means for $\forall j,\,\exists i$
such that $J_{j}\subset I_{i}$.
Let $(\tilde{I}_{i})=(I_{i})\diagup (J_{j})$,
$I=\mathcal{R}(\{I_{i}\})$ and $J=\mathcal{R}(\{J_{j}\})$.
Then $(\tilde{I}_{i})=(I_{i^{\prime}})
\cup(I_{i^{\prime\prime}}\diagup(J_{j}))$,
where $I_{i^{\prime}}\cap J=\emptyset$,
$I_{i^{\prime\prime}}\cap J\not=\emptyset$,
and we identify $I_{i^{\prime}}\cup\{\emptyset\}$
with $I_{i^{\prime}}$. Now we assume 
$I_{i}\cap J\not=\emptyset$ for any $i$.
There is a decomposition
of the partition $(J_{j})$

$$
(J_{j})=\bigcup\limits_{i=1}^{l}J^{(i)},\,
J^{(i)}=(J_{\lambda_{ij}})\subset I_{i},\,i=1,\cdots,l.
$$
Then we have 

$$
\tilde{I}_{i}=I_{i}\diagup (J_{j})_{J_{j}\subset I_{i}}
=I_{i}\diagup (J_{\lambda_{ij}})
=(I_{i}\setminus J)\cup\{J_{\lambda_{ij}}\}.
$$

Now we turn to a more complicated situation. 
At first, we introduce two notations:

\begin{itemize}
	\item  $\mathbf{\mathcal{P}_{dis}^{2}(A)\times_{dis}
		\mathcal{P}_{dis}(\Xi_{A}):}$ It is 
	a subset of $\mathcal{P}_{dis}^{2}(A)\times
	\mathcal{P}_{dis}(\Xi_{A})$.
	An elment 
	$(\{I_{i}\},\{L_{j}\})\in\mathcal{P}_{dis}^{2}(A)\times_{dis}
	\mathcal{P}_{dis}(\Xi_{A})$ means
	$(I_{i})\cap(\mathcal{R}_{1}(\{L_{j}\}))=\emptyset$,
	where $\{I_{i}\}\in \mathcal{P}_{dis}^{2}(A)$ and
	$\{L_{j}\}\in \mathcal{P}_{dis}(\Xi_{A})$.
	Generally, we do not distinguish $(\{I_{i}\},\{L_{j}\})$
	from $\{I_{i}\}\cup\{L_{j}\}$, because $\{I_{i}\}$ and
	$\{L_{j}\}$ are elements in the sets of completely 
	different types.
	
	\item $\mathbf{\Xi_{A}\times_{dis} \mathcal{P}_{dis}(\Xi_{A}):}$ 
	It is a subset of
	$\Xi_{A}\times \mathcal{P}_{dis}(\Xi_{A})$.
	An element $((\{U\},\{K_{\lambda}\}),\{I_{i}\cup J_{i}\})$ in
	$\Xi_{A}\times_{dis} \mathcal{P}_{dis}(\Xi_{A})$
	denoted by 
	$U\cup\{K_{\lambda}\}\cup\{I_{i}\cup J_{i}\}$ also, 
	where $\{U\}\in \mathcal{P}(A)$,
	$\{K_{\lambda}\}\in \mathcal{P}^{2}_{dis}(A)$ and 
	$\{I_{i}\cup J_{i}\}\in \mathcal{P}_{dis}(\Xi_{A})$, 
	satisfies
	
	$$
	\begin{array}{c}
	\{K_{\lambda}\}\cup\{I_{i}\cup J_{i}\}\in
	\mathcal{P}_{dis}^{2}(A)
	\times_{dis}\mathcal{P}_{dis}(\Xi_{A}) \\    
	U\cap \mathcal{R}
	(\{K_{\lambda}\}\cup\{I_{i}\cup\mathcal{R}_{1}(\{J_{i}\})\})=\emptyset.
	\end{array}
	$$
\end{itemize}

Let
$U\subset A$, $\{K_{\lambda}\}\in\mathcal{P}^{2}_{dis}(A)$,
$(K_{\lambda})\subset U$, $\{I_{i}\cup J_{i}\}
\in\mathcal{P}_{dis}(\Xi_{A})$,
$(I_{i}\cup J_{i})\subset U\diagup(K_{\lambda})$.
Let $I=\mathcal{R}(\{I_{i}\})$,
$J=\mathcal{R}(\{J_{i}\})$, 
$K=\mathcal{R}(\{K_{\lambda}\})$.
Noting

$$
U\diagup(K_{\lambda})=
(U\setminus K)\cup\{K_{\lambda}\},
$$
thus the quotient of a set by partition
is a map

$$
\mathcal{Q}_{1}:\mathcal{P}(A)\times\mathcal{P}_{dis}^{2}(A)
\longrightarrow\Xi_{A}.
$$
Recalling the contents in section 2.1,
we know that
$(I_{i}\cup J_{i})\subset U\diagup(K_{\lambda})$
means that $\mathcal{R}(\{I_{i}\})\subset 
U\setminus\mathcal{R}(\{K_{\lambda}\})$, and $(I_{i})$
is a partition in $U\setminus\mathcal{R}(\{K_{\lambda}\})$,
moreover, $(J_{i})$ is a partition in $\{K_{\lambda}\}$,
where $\{K_{\lambda}\}$ is regarded as a subset
in $\mathcal{P}_{dis}^{2}(A)$.
Thus, we have

$$
\begin{array}{c}
((U\diagup(K_{\lambda}))\diagup(I_{i}\cup J_{i}) \\
=((U\setminus K)\cup\{K_{\lambda}\})\diagup(I_{i})\cup(J_{i})
=((U\setminus K)\diagup(I_{i}))\cup
(\{K_{\lambda}\}\diagup(J_{i})) \\
=(U\setminus(K\cup I))\cup
\{I_{i}\}\cup(\{K_{\lambda}\}\setminus J)\cup\{J_{i}\} \\
=(U\setminus(K\cup I))
\cup(\{K_{\lambda}\}\setminus J)\cup\{I_{i}\cup J_{i}\}.
\end{array}
$$
The quotient $((U\setminus K)
\cup\{K_{\lambda}\})\diagup(I_{i}\cup J_{i})$
is a map

$$
\mathcal{Q}_{2}:\Xi_{A}\times \mathcal{P}_{dis}(\Xi_{A})
\longrightarrow
\Xi_{A}\times_{dis}\mathcal{P}_{dis}(\Xi_{A}).
$$

There is a separation $\{I_{i}\cup J_{i}\}=\{I_{i^{\prime}}\cup\emptyset\}
\cup\{I_{i^{\prime\prime}}\cup J_{i^{\prime\prime}}\}$,
where $J_{i^{\prime\prime}}\neq\emptyset$. 
It is obvious that 
$$
(U\diagup(K_{\lambda}))\diagup(I_{i}\cup J_{i})=
(U\diagup((K_{\lambda})\cup(I_{i^{\prime}})))
\diagup(I_{i^{\prime\prime}}\cup J_{i^{\prime\prime}})
$$
if we identify $\{I_{i^{\prime}}\}$ with
$\{I_{i^{\prime}}\cup\emptyset\}$.

\begin{remark}
In general, for $U\subset A$, 
$\{K_{\lambda}\}\in\mathcal{P}^{2}_{dis}(A)$,
$(K_{\lambda})\subset U$, $\{I_{i}\cup J_{i}\}
\in\mathcal{P}_{dis}(\Xi_{A})$,
we define 

$$
(U\diagup(K_{\lambda}))\diagup(I_{i}\cup J_{i})
\doteq(U\diagup(K_{\lambda}))\diagup(I_{i}\cup J_{i})
_{I_{i}\cup J_{i}\subset U\diagup(K_{\lambda})}.
$$
\end{remark}

\begin{definition}
An element $U\cup\{K_{\lambda}\}\cup\{I_{i}\cup J_{i}\}
\in\Xi_{A}\times_{dis}\mathcal{P}_{dis}(\Xi_{A})$
will be devided
into two parts according to the power degree:

\begin{itemize}
	\item We call $\{U\}$ (or $U$ simply) the \textbf{original part}
	of $U\cup\{K_{\lambda}\}\cup\{I_{i}\cup J_{i}\}$ denoted by
	$$
	[U\cup\{K_{\lambda}\}\cup\{I_{i}\cup J_{i}\}]_{or}.
	$$ 
	\item We call $\{K_{\lambda}\}\cup\{I_{i}\cup J_{i}\}$
	the \textbf{ideal part} of $U\cup\{K_{\lambda}\}\cup\{I_{i}\cup J_{i}\}$
	denoted by 
	$$
	[U\cup\{K_{\lambda}\}\cup\{I_{i}\cup J_{i}\}]_{id}.
	$$
\end{itemize}
\end{definition}

\subsection{Reversion map and induced quotient}

\paragraph{Reversion map}

In this subsection we discuss the issues concerning 
the reversion map and the power degree in details.
Here we centre on the situation of quotient by partitions.
The reversion map introduced in subsection 2.1
describes the inverse procedure of quotient and
decreases the power degree of a set. 
It is obvious that

$$
\mathcal{R}(\{I_{i}\})\diagup (I_{i})=\emptyset\cup\{I_{i}\}.
$$
Generally, the quotient of a subset in $A$ by a partition
is the map from $\mathcal{P}(A)\times\mathcal{P}^{2}_{dis}(A)$
to $\Xi_{A}$. Conversely, for a given 
$(\{U\},\{I_{i}\})\in \Xi_{A}$ we have

$$
(U\cup\mathcal{R}(\{I_{i}\}))\diagup (I_{i})
=U\cup(\mathcal{R}(\{I_{i}\}\diagup (I_{i})))=U\cup\{I_{i}\}.
$$
Moreover, we have the following lemma:

\begin{lemma}(\textbf{The uniqueness of "molecule"}) 
Let $V\subset A$, $\{U\}\cup\{I_{i}\}\in\Xi_{A}$,
we have
$$
V\diagup(I_{i})=U\cup\{I_{i}\}\,\Longleftrightarrow
V=U\cup\mathcal{R}(\{I_{i}\}).
$$
\end{lemma}

\begin{proof}
By definition of quotient we have
$$
V\diagup(I_{i})=(V\setminus\mathcal{R}(\{I_{i}\}))
\cup\{V\cap I_{i}\}.
$$
Noting $\{V\cap I_{i}\}=\{I_{i}\}$, we know that
$I_{i}\subset V$, thus, $(I_{i})\subset V$.
On the other hand $V\setminus\mathcal{R}(\{I_{i}\})
=U$, so the conclusion of lemma has been proved.
\end{proof}

In general, we have a generalization of lemma 2.1.

\begin{lemma}
	Let $U\cup\{K_{\lambda}\}\cup\{I_{i}\cup J_{i}\}
	\in\Xi_{A}\times_{dis} \mathcal{P}_{dis}(\Xi_{A})$,
	$I=\mathcal{R}(\{I_{i}\})$, $J=\mathcal{R}(\{J_{i}\})$,
	$K=\mathcal{R}(\{K_{\lambda}\})$, we have
	\begin{equation}
	\begin{array}{c}
	U\cup\{K_{\lambda}\}\cup\{I_{i}\cup J_{i}\} \\
	=((U\cup I\cup K\cup\mathcal{R}(J))
	\diagup((K_{\lambda})\cup J))\diagup
	(I_{i}\cup J_{i}).
	\end{array}
	\end{equation}
\end{lemma}

\begin{proof}
Let $I=\mathcal{R}(\{I_{i}\})$ and
$J=\mathcal{R}(\{J_{i}\})$, we have
$$
I\cup J=\mathcal{R}(\{I_{i}\cup J_{i}\})
=\mathcal{R}(\{I_{i}\})\cup\mathcal{R}(\{J_{i}\}).
$$
Then we have
$$
\begin{array}{c}
U\cup\{K_{\lambda}\}\cup\{I_{i}\cup J_{i}\} \\
=((U\cup I\cup\{K_{\lambda}\}\cup J)\setminus (I\cup J))
\cup\{I_{i}\cup J_{i}\} \\
=(U\cup I\cup\{K_{\lambda}\}\cup J)\diagup(I_{i}\cup J_{i}).
\end{array}
$$
Furthermore, let $K=\mathcal{R}(\{K_{\lambda}\})$,
noting that
$$
\begin{array}{c}
U\cup I\cup\{K_{\lambda}\}\cup J \\
=((U\cup I\cup K\cup\mathcal{R}(J))
\setminus(K\cup\mathcal{R}(J)))
\cup(\{K_{\lambda}\}\cup J) \\
=(U\cup I\cup K\cup\mathcal{R}(J))
\diagup((K_{\lambda})\cup J),
\end{array}
$$

\end{proof}

We need to pay attention to an interesting fact:

\begin{lemma}
	Let $\{I_{i}\}\in\mathcal{P}_{dis}^{2}(A)$, 
	$\{J_{j}\}=\{J_{j}^{\prime}\cup 
	J_{j}^{\prime\prime}\}\in\mathcal{P}_{dis}(\Xi_{A})$,
	where $(J_{j}^{\prime\prime})$ is a partition
	in $\{I_{i}\}$,  
	then
	
	\begin{equation}
	(J_{j})=(\mathcal{R}_{1}(J_{j}))
	\diagup(I_{i})_{I_{i}\in
		\mathcal{R}(\{J_{j}^{\prime\prime}\})}.
	\end{equation}
\end{lemma}

\begin{proof}
	We know that $(J_{j}^{\prime\prime})$ is
	a partition in $\{I_{i}\}$. Thus each 
	$J_{j}^{\prime\prime}$ is subset of $\{I_{i}\}$
	which means $J_{j}^{\prime\prime}=
	\{I_{i}\}_{I_{i}\in J_{j}^{\prime\prime}}$.
	Then we have
	
	$$
	J_{j}^{\prime\prime}=
	\mathcal{R}(J_{j}^{\prime\prime})
	\diagup J_{j}^{\prime\prime}=
	\mathcal{R}(J_{j}^{\prime\prime})\diagup
	(I_{i})_{I_{i}\in J_{j}^{\prime\prime}}.
	$$
	Furthermore, 
	
	$$
	\begin{array}{c}
	J_{j}=J_{j}^{\prime}\cup J_{j}^{\prime\prime}
	=J_{j}^{\prime}\cup
	(\mathcal{R}(J_{j}^{\prime\prime})\diagup
	(I_{i})_{I_{i}\in J_{j}^{\prime\prime}}) \\
	=(J_{j}^{\prime}\cup\mathcal{R}(J_{j}^{\prime\prime}))
	\diagup(I_{i})_{I_{i}\in J_{j}^{\prime\prime}}
	=\mathcal{R}_{1}(J_{j})\diagup(I_{i})_{I_{i}\in J_{j}^{\prime\prime}}.
	\end{array}
	$$
	Above formula implies (2.10) immediately.

\end{proof}

\paragraph{Induced quotient, case of $(U\diagup (I_{i}))\diagup (J_{j})$:}

We are really interested in the case of 
$(U\diagup (I_{i}))\diagup (J_{j})$
which is same as the situation discussed
in the previous subsection with slightly
different form,
where $U\subset A$, $(I_{i})$ is a partition in $U$ and
$(J_{j})$ is a partition in $U\diagup (I_{i})$.
It is obvious that we have

$$
U\diagup(I_{i})=
(U\setminus I)\cup\{I_{i}\},
$$
where $I=\mathcal{R}(\{I_{i}\})$.
We now focus on the following quotient:

$$
((U\setminus I)\cup\{I_{i}\})
\diagup(J_{j}).
$$
$\{J_{j}\}$ endows with a decomposition
$J_{j}=J_{j}^{\prime}\cup J_{j}^{\prime\prime}$,
where $J_{j}^{\prime}=J_{j}\cap(U\setminus I)$,
$J_{j}^{\prime\prime}=J_{}\cap\{I_{i}\}$.
Thus $(J_{j})\in\mathcal{P}_{dis}(\Xi_{A})$
and 

$$
(U\diagup (I_{i}))\diagup (J_{j})=
(U\setminus (I\cup J^{\prime}))\cup
\{I_{i}|I_{i}\notin J\}\cup \{J_{j}\}, 
$$ 
where $J=\mathcal{R}(\{J_{j}\})$, and
$J^{\prime}=J\cap(U\setminus I)
=\mathcal{R}(\{J_{j}^{\prime}\})$.

By definition 2.2,
$(U\diagup (I_{i}))\diagup (J_{j})$ can be divided
into two parts as following:

\begin{itemize}
\item

$$
[(U\diagup (I_{i}))\diagup (J_{j})]_{or}=
U\setminus (I\cup J^{\prime})\in \mathcal{P}(A),
$$
\item

$$
[(U\diagup (I_{i}))\diagup (J_{j})]_{id}=
\{I_{i}|I_{i}\notin J\}\cup \{J_{j}\}\in
\mathcal{P}_{dis}^{2}(A)\times_{dis}
\mathcal{P}_{dis}(\Xi_{A}).
$$
\end{itemize}

\begin{remark}
In general, $\{I_{i}\}\setminus
J^{\prime\prime}\neq\emptyset$, 
where $J^{\prime\prime}
=\mathcal{R}(\{J_{j}^{\prime\prime}\})$.
The subset $\{I_{i}|I_{i}\notin J^{\prime\prime}\}$ of
$[(U\diagup (I_{i}))\diagup (J_{j})]_{id}$ measures
the difference between $\{I_{i}\}$
and $\mathcal{R}(\{J_{j}\})$.
\end{remark}

We dived the partition $(I_{i})$ into two parts,
$(I_{i})=(I_{i^{\prime}})\cup(I_{i^{\prime\prime}})$,
where $(I_{i^{\prime}})=(I_{i})_{I_{i}\notin J^{\prime\prime}},\,          
(I_{i^{\prime\prime}})=(I_{i})_{I_{i}\in J^{\prime\prime}}$,
(or $\{I_{i^{\prime\prime}}\}=J^{\prime\prime}$).
Considering
$$
\begin{array}{c}
(id\otimes\mathcal{R}_{1})
([(U\diagup (I_{i}))\diagup (J_{j})]_{id}) \\
=\{I_{i^{\prime}}\}
\cup \{\mathcal{R}_{1}(J_{j})\}
=\{I_{i^{\prime}}\}\cup
\{J_{j}^{\prime}\cup\mathcal{R}(J_{j}^{\prime\prime})\},
\end{array}
$$
let 
$$
\{K_{\lambda}\}=(id\otimes\mathcal{R}_{1})
([(U\diagup (I_{i}))\diagup (J_{j})]_{id}),
$$
then
$$
K=\mathcal{R}(\{K_{\lambda}\})=I\cup J^{\prime},\,
(I_{i})\subset (K_{\lambda}),\, (K_{\lambda})=(I_{i^{\prime}})
\cup(\mathcal{R}_{1}(J_{j})),
$$
It is easy to check that

$$
U\diagup (K_{\lambda})=
(U\diagup(I_{i^{\prime}}))\diagup(\mathcal{R}(J_{j})).         
$$
In summary, we have

\begin{proposition}
Let $(I_{i})$ be a partition in $U$, $U\subset A$,
$(J_{j})$ be a partition in $U\diagup (I_{i})$, 
$(I_{i})=(I_{i^{\prime}})\cup(I_{i^{\prime\prime}})$
as above, taking

\begin{equation}
\{K_{\lambda}\}=(id\otimes\mathcal{R}_{1})
([(U\diagup (I_{i}))\diagup (J_{j})]_{2})
=\{I_{i^{\prime}}\}\cup\{\mathcal{R}_{1}(J_{j})\},
\end{equation}
then we have

\begin{equation}
U\diagup (K_{\lambda})=
(U\diagup(I_{i^{\prime}}))\diagup(\mathcal{R}_{1}(J_{j})),
\end{equation} 
\end{proposition}

\begin{definition}
We call $U\diagup (K_{\lambda})$ in (2.12)
the \textbf{induced quotient}
of $((U\diagup (I_{i}))\diagup(J_{j})$
denoted by
\begin{equation}
U\diagup (K_{\lambda})=\textit{ind}
\{((U\diagup (I_{i}))\diagup(J_{j})\},
\end{equation}
where $(K_{\lambda})$ is given by (2.11).
\end{definition}

\begin{remark}
In a trivial situation of 
$J_{j}^{\prime\prime}=\emptyset$
for all $j$, equivalently, 
$(J_{j})\subset U\setminus I$, 
we have

$$
ind\{(U\diagup(I_{i}))\diagup(J_{j})\}=
(U\diagup(I_{i}))\diagup(J_{j})=
U\diagup((I_{i})\cup(J_{j})).
$$ 
\end{remark}

Roughly speaking, induced quotient arises
from decreasing of the power degree.
Comparing with proposition 2.1 and lemma 2.3, 
we get the following properties of the
induced quotient.

\begin{corollary}
	
$$
\textit{ind}\{((U\diagup(I_{i^{\prime}}))\diagup(I_{i^{\prime\prime}}))
\diagup((\mathcal{R}_{1}(J_{j}))\diagup(I_{i^{\prime\prime}}))\}
=(U\diagup(I_{i^{\prime}}))\diagup(\mathcal{R}_{1}(J_{j})).
$$ 
Particularly, if 
$\{I_{i}\}=\mathcal{R}(\{J_{j}^{\prime\prime}\})$,
where $J_{j}^{\prime\prime}=J_{j}\cap\{I_{i}\}$,
we have

$$
\textit{ind}\{(U\diagup(I_{i}))\diagup((\mathcal{R}_{1}(J_{j}))\diagup(I_{i}))\}
=U\diagup(\mathcal{R}_{1}(J_{j})).
$$
Furthermore, for two partitions in $U$, 
$(I_{i})\subset(K_{\lambda})\subset U$, we have

$$
\textit{ind}\{(U\diagup(I_{i}))\diagup((K_{\lambda})\diagup(I_{i}))\}
=U\diagup(K_{\lambda}).
$$
\end{corollary}

\paragraph{Case of $U\cup\{I_{i}\}
\diagup(K_{\lambda}\cup L_{\lambda})$:}

Now w discuss the induced quotient starting from
$\Xi_{A}$ and $\mathcal{P}_{dis}(\Xi_{A})$. 
Let $\{U\}\cup\{I_{i}\}\in\Xi_{A}$,
$\{K_{\lambda}\cup L_{\lambda}\}\in \mathcal{P}_{dis}(\Xi_{A})$,
$(K_{\lambda}\cup L_{\lambda})\subset U\cup\{I_{i}\}$,
then we have
$$
U\cup\{I_{i}\}\diagup(K_{\lambda}\cup L_{\lambda})=
(U\setminus\mathcal{R}(\{K_{\lambda}\}))\cup
(\{I_{i}\}\setminus\mathcal{R}(\{L_{\lambda}\}))\cup
\{K_{\lambda}\cup L_{\lambda}\}.
$$
The ideal part of $U\cup\{I_{i}\}
\diagup(K_{\lambda}\cup L_{\lambda})$ is

$$
[U\cup\{I_{i}\}\diagup(K_{\lambda}\cup L_{\lambda})]_{id}=
(\{I_{i}\}\setminus\mathcal{R}(\{L_{\lambda}\}))\cup
\{K_{\lambda}\cup L_{\lambda}\}.
$$
To calculate the induced quotient of
$U\cup\{I_{i}\}\diagup(K_{\lambda}\cup L_{\lambda})$
we take reversion,

$$
(id\times\mathcal{R}_{1})
([U\cup\{I_{i}\}\diagup(K_{\lambda}\cup L_{\lambda})]_{id}) 
=(\{I_{i}\}\setminus\mathcal{R}(\{L_{\lambda}\}))\cup
\{K_{\lambda}\cup \mathcal{R}(L_{\lambda})\}.   
$$
Let 

$$
\{J_{j}\}=(id\times\mathcal{R}_{1})
((\{I_{i}\}\diagup(K_{\lambda}\cup L_{\lambda}))_{id}),
$$
or,

$$
(J_{j})=(I_{i})_{I_{i}\notin\mathcal{R}(\{L_{\lambda}\})}\cup
(K_{\lambda}\cup \mathcal{R}(L_{\lambda})).
$$
Due to the facts which are $(L_{\lambda})$ is a partition
in $\{I_{i}\}$, thus, each $L_{\lambda}$ is a subset
of $\{I_{i}\}$, we can get the following facts:
\begin{itemize}
\item
$$
(K_{\lambda}\cup L_{\lambda})=
(K_{\lambda}\cup \mathcal{R}(L_{\lambda}))
\diagup(I_{i})_{I_{i}\in\mathcal{R}(\{L_{\lambda}\})},
$$
\item
$$
\mathcal{R}(\{J_{j}\})=\mathcal{R}(\{I_{i}\})
\cup\mathcal{R}(\{K_{\lambda}\}),
$$
\item
$$
(U\cup\mathcal{R}(\{I_{i}\}))\diagup(J_{j})
=(U\setminus\mathcal{R}(\{K_{\lambda}\}))
\cup\{J_{j}\}.
$$
\end{itemize}
Therefore we have

\begin{proposition}
	
\begin{equation}
\textit{ind}\{U\cup\{I_{i}\}\diagup(K_{\lambda}\cup L_{\lambda})\}
=(U\cup\mathcal{R}(\{I_{i}\}))\diagup(J_{j}),
\end{equation}
where

$$
(J_{j})=(I_{i})_{I_{i}\notin\mathcal{R}(\{L_{\lambda}\})}\cup
(K_{\lambda}\cup \mathcal{R}(L_{\lambda})).
$$
\end{proposition}

\paragraph{Case of $(I_{i}\cup J_{i})
\diagup (K_{\lambda}\cup L_{\lambda})$:}

Now we consider more complex case.
Let $\{K_{\lambda}\cup L_{\lambda}\},\,\{I_{i}\cup J_{i}\}
\in\mathcal{P}_{dis}(\Xi_{A})$, and
$(K_{\lambda}\cup L_{\lambda})\subset(I_{i}\cup J_{i})$
which means for $\forall\lambda$, $\exists i$ such that
$K_{\lambda}\cup L_{\lambda}\subset I_{i}\cup J_{i}$.
Now we want to discuss the induced quotient in the situation of
$(I_{i}\cup J_{i})\diagup (K_{\lambda}\cup L_{\lambda})$.
At first we know that there is decomposition of partition
$(K_{\lambda}\cup L_{\lambda})$
$$
(K_{\lambda}\cup L_{\lambda})=\bigcup\limits_{i}
(K_{\lambda_{ij}}\cup L_{\lambda_{ij}}),\,
(K_{\lambda_{ij}}\cup L_{\lambda_{ij}})\subset I_{i}\cup J_{i}.
$$
Thus $(K_{\lambda_{ij}})$ is a partition in $I_{i}$,
and $(L_{\lambda_{ij}})$ is a partition in $J_{i}$ also.
Let $K^{(i)}$ denote $(K_{\lambda_{ij}})$ 
(or $\{K_{\lambda_{ij}}\}$) and $L^{(i)}$ denote
$(L_{\lambda_{ij}})$ (or $\{L_{\lambda_{ij}}\}$).
For simplity, we assume $L^{(i)}\not=\emptyset$
for all $i$.
Similar to above discussion we have

$$
(I_{i}\cup J_{i})\diagup (K_{\lambda}\cup L_{\lambda})
=([I_{i}\cup J_{i}]\diagup (K_{\lambda_{ij}}\cup L_{\lambda_{ij}})),
$$
where $(I_{i}\cup J_{i})\diagup (K_{\lambda}\cup L_{\lambda})$
means the quotient of partition $(I_{i}\cup J_{i})$
by partition $(K_{\lambda}\cup L_{\lambda})$, and
$[I_{i}\cup J_{i}]\diagup (K_{\lambda_{ij}}\cup L_{\lambda_{ij}})$
means the quotient of set $I_{i}\cup J_{i}$
by partition $(K_{\lambda_{ij}}\cup L_{\lambda_{ij}})$.
It is natural for us to define the induced quotient
in this situation to be

\begin{equation}
ind\{(I_{i}\cup J_{i})\diagup (K_{\lambda}\cup L_{\lambda})\}=
(ind\{[I_{i}\cup J_{i}]\diagup (K_{\lambda_{ij}}\cup L_{\lambda_{ij}})\}).
\end{equation}

$$
\begin{array}{c}
[I_{i}\cup J_{i}]\diagup (K_{\lambda_{ij}}\cup L_{\lambda_{ij}}) \\
=(I_{i}\setminus\mathcal{R}(K^{(i)}))\cup
(J_{i}\setminus\mathcal{R}(L^{(i)}))\cup
\{K_{\lambda_{ij}}\cup L_{\lambda_{ij}}\} \\
=(I_{i}\setminus K)\cup(J_{i}\setminus L)\cup
\{K_{\lambda_{ij}}\cup L_{\lambda_{ij}}\},
\end{array}
$$
where $K=\bigcup_{\lambda}K_{\lambda}$,
$L=\bigcup_{\lambda}L_{\lambda}$. Then we have

$$
(I_{i}\cup J_{i})\diagup (K_{\lambda}\cup L_{\lambda})
=((I_{i}\setminus K)\cup(J_{i}\setminus L)\cup
\{K_{\lambda_{ij}}\cup L_{\lambda_{ij}}\}).
$$
Let

$$
M_{\lambda_{ij}}=(id\times\mathcal{R}_{1})([I_{i}\cup J_{i}\diagup 
(K_{\lambda_{ij}}\cup L_{\lambda_{ij}})]_{id})=
(J_{i}\setminus L)\cup
\{K_{\lambda_{ij}}\cup\mathcal{R}(L_{\lambda_{ij}})\},
$$
by proposition 2.2 we know that

$$
ind\{[I_{i}\cup J_{i}]\diagup 
(K_{\lambda_{ij}}\cup L_{\lambda_{ij}})\}=
[I_{i}\cup\mathcal{R}(J_{i})]\diagup(M_{\lambda_{ij}}).
$$
Now we take

$$
(M_{\mu})=\bigcup_{i}(M_{\lambda_{ij}})=
(J_{i}\setminus L)\cup
(K_{\lambda}\cup\mathcal{R}(L_{\lambda})),
$$
then we have

\begin{proposition}
	
\begin{equation}
ind\{(I_{i}\cup J_{i})\diagup (K_{\lambda}\cup L_{\lambda})\}=
(I_{i}\cup\mathcal{R}(J_{i}))\diagup(M_{\mu}),
\end{equation}
where 

$$
(M_{\mu})=(J_{i}\setminus L)\cup
(K_{\lambda}\cup\mathcal{R}(L_{\lambda})).
$$
\end{proposition}

Particularly, if $\mathcal{R}(\{J_{i}\})
=\mathcal{R}(\{L_{\lambda}\})$, then
 
$$
\begin{array}{c}
(M_{\mu})=(K_{\lambda}\cup \mathcal{R}(L_{\lambda}),\,
\mathcal{R}(\{J_{i}\})\subset(K_{\lambda}\cup\mathcal{R}(L_{\lambda})), \\
(K_{\lambda}\cup L_{\lambda})=(K_{\lambda}\cup\mathcal{R}(L_{\lambda})) 
\diagup(\mathcal{R}(\{L_{\lambda}\}))=
(K_{\lambda}\cup\mathcal{R}(L_{\lambda})) 
\diagup(\mathcal{R}(\{J_{i}\})),
\end{array}
$$ 
where $L_{\lambda}\in\mathcal{P}_{dis}^{2}(A)$,
thus $\mathcal{R}(L_{\lambda})\in\mathcal{P}(A)$
and $\mathcal{R}(\{L_{\lambda}\})\in\mathcal{P}_{dis}^{2}(A)$,
so is $\mathcal{R}(\{J_{i}\})$.
Finally, we reach the following conclusions,

\begin{corollary}
Let $\{K_{\lambda}\cup L_{\lambda}\},\,\{I_{i}\cup J_{i}\}
\in\mathcal{P}_{dis}(\Xi_{A})$,
$(K_{\lambda}\cup L_{\lambda})\subset(I_{i}\cup J_{i})$,
$\mathcal{R}(\{J_{i}\})=\mathcal{R}(\{L_{\lambda}\})$. Then

$$
\begin{array}{c}
\textit{ind}\{(I_{i}\cup J_{i})\diagup (K_{\lambda}\cup L_{\lambda})\} \\
=\textit{ind}\{((I_{i}\cup\mathcal{R}(J_{i}))\diagup(\mathcal{R}(\{J_{i}\})))\diagup
((K_{\lambda}\cup\mathcal{R}(L_{\lambda}))\diagup(\mathcal{R}(\{J_{i}\})))\} \\
=(I_{i}\cup\mathcal{R}(J_{i}))\diagup(K_{\lambda}\cup \mathcal{R}(L_{\lambda})).
\end{array}
$$
\end{corollary}

\begin{corollary} 
Let $(I_{i})\subset(W_{\mu})\subset(V_{\lambda})\subset U$,
$(K_{\nu})=(W_{\mu})\diagup(I_{i})$, 
$(J_{j})=(V_{\lambda})\diagup(I_{i})$,
then we have
\begin{equation}
\textit{ind}\{(J_{j})\diagup(K_{\nu})\}=(V_{\lambda})\diagup(W_{\mu}).
\end{equation}
\end{corollary}

Furthermore, the induced quotient has
the following property:
\begin{proposition}
Let $\{U\}\cup\{I_{i}\}\in\Xi_{A}$, $\{D_{j}\cup E_{j}\},\,
\{K_{\lambda}\cup L_{\lambda}\}\in\mathcal{P}_{dis}(\Xi_{A})$,
$(D_{j}\cup E_{j})\subset(K_{\lambda}\cup L_{\lambda})
\subset U\cup\{I_{i}\}$, we have
\begin{equation}
\begin{array}{c}
\textit{ind}\{U\cup\{I_{i}\}\diagup(K_{\lambda}\cup L_{\lambda})\} \\
=\textit{ind}\{\textit{ind}\{U\cup\{I_{i}\}\diagup(D_{j}\cup E_{j})\}\diagup
\textit{ind}\{(K_{\lambda}\cup L_{\lambda})\diagup
(D_{j}\cup E_{j})\}\}.
\end{array}
\end{equation}
\end{proposition}
\begin{proof}
Recalling the previous discussion we know that
$$
\begin{array}{c}
\textit{ind}\{U\cup\{I_{i}\}\diagup(D_{j}\cup E_{j})\} \\
=(U\setminus\mathcal{R}(\{D_{j}\})\cup
(\{I_{i}\}\setminus\mathcal{R}(\{E_{j}\}))\cup
\{D_{j}\cup\mathcal{R}(E_{j})\},
\end{array}
$$
and
$$
\begin{array}{c}
\textit{ind}\{K_{\lambda}\cup L_{\lambda}\diagup
(D_{j}\cup E_{j})\} \\
=(K_{\lambda}\setminus\mathcal{R}(\{D_{j}^{(\lambda)}\}))\cup
(L_{\lambda}\setminus\mathcal{R}(\{E_{j}^{(\lambda)}\}))\cup
\{D_{j}^{(\lambda)}\cup\mathcal{R}(E_{j}^{(\lambda)})\},  
\end{array}
$$
where $(D_{j}\cup E_{j})=\bigcup_{\lambda}
(D_{j}^{(\lambda)}\cup E_{j}^{(\lambda)})$,
$(D_{j}^{(\lambda)}\cup E_{j}^{(\lambda)})\subset
K_{\lambda}\cup L_{\lambda}$.
Then we have
$$
\begin{array}{c}
\textit{ind}\{U\cup\{I_{i}\}\diagup(D_{j}\cup E_{j})\}\diagup
\textit{ind}\{K_{\lambda}\cup L_{\lambda}\diagup(D_{j}\cup E_{j})\} \\
=(U\setminus\mathcal{R}(\{K_{\lambda}\}))\cup
(\{I_{i}\}\setminus\mathcal{R}(\{L_{\lambda}\}))\cup\{M_{\lambda}\},  
\end{array}
$$
where
$$
M_{\lambda}=(K_{\lambda}\setminus\mathcal{R}(\{D_{j}^{(\lambda)}\}))\cup
(L_{\lambda}\setminus\mathcal{R}(\{E_{j}^{(\lambda)}\}))\cup
\{D_{j}^{(\lambda)}\cup\mathcal{R}(E_{j}^{(\lambda)})\}.
$$
Now we make reversion for $M_{\lambda}$ which is
$$
\mathcal{R}_{1}(M_{\lambda})=
(K_{\lambda}\setminus\mathcal{R}(\{D_{j}^{(\lambda)}\}))\cup
(\mathcal{R}(L_{\lambda})\setminus\mathcal{R}(\mathcal{R}(\{E_{j}^{(\lambda)}\})))
\cup\mathcal{R}(\{D_{j}^{(\lambda)}\})\cup
\mathcal{R}(\{\mathcal{R}(E_{j}^{(\lambda)})\}). 
$$
Noting that
$$
\mathcal{R}(\mathcal{R}(\{E_{j}^{(\lambda)}\}))
=\mathcal{R}(\{\mathcal{R}(E_{j}^{(\lambda)})\}),
$$
we have
$$
\mathcal{R}_{1}(M_{\lambda})=
K_{\lambda}\cup \mathcal{R}(L_{\lambda}).
$$
Finally, we have
$$
\begin{array}{c}
\textit{ind}\{\textit{ind}\{U\cup\{I_{i}\}\diagup(D_{j}\cup E_{j})\}\diagup
\textit{ind}\{(K_{\lambda}\cup L_{\lambda})\diagup
(D_{j}\cup E_{j})\}\} \\
=(U\setminus\mathcal{R}(\{K_{\lambda}\}))\cup
(\{I_{i}\}\setminus\mathcal{R}(\{L_{\lambda}\}))\cup
\{K_{\lambda}\cup \mathcal{R}(L_{\lambda})\}.
\end{array}
$$
Thus the formula (2.18) is valid.
\end{proof}

\section{Hopf algebra consisting of finite sets}

In this section we want to construct the 
coproduct for finite sets. Let $A$ be a finite set. 
We will construct a coproduct
$$
\triangle:\mathcal{P}_{dis}(\Xi_{A})
\longrightarrow \mathcal{P}_{dis}(\Xi_{A})\otimes 
\mathcal{P}_{dis}(\Xi_{A}).
$$
At first we construct sub-coproduct related to a partition 
as follows:

\begin{definition}
$\\$
\begin{itemize}
\item Let $U\cup\{I_{i}\}\in\Xi_{A}$, 
$\{K_{\lambda}\cup L_{\lambda}\}
\in\mathcal{P}_{dis}(\Xi_{A})$, 
$(K_{\lambda}\cup L_{\lambda})\subset U\cup\{I_{i}\}$,
then we define
\begin{equation}
\begin{array}{c}
\triangle_{(K_{\lambda}\cup L_{\lambda})}(U\cup\{I_{i}\}) \\
=(K_{\lambda}\cup L_{\lambda})\otimes \textit{ind}\{U\cup\{I_{i}\}\diagup
(K_{\lambda}\cup L_{\lambda})\}.
\end{array}
\end{equation}
Where $(K_{\lambda}\cup L_{\lambda})$ satisfies the
following conditions:
\begin{itemize}
  \item $\mathcal{R}(\{L_{\lambda}\})=\{I_{i}\}$,
  for each $\lambda$, 
  $L_{\lambda}\neq\{I_{i}\}$.  
  \item For each $\lambda$,
  $K_{\lambda}\neq U$.
\end{itemize}
\item  Let $\{K_{\lambda}\cup L_{\lambda}\},\,\{I_{i}\cup J_{i}\}
\in\mathcal{P}_{dis}(\Xi_{A})$, 
$(K_{\lambda}\cup L_{\lambda})\subset(I_{i}\cup J_{i})$, 
$\mathcal{R}(\{L_{\lambda}\})=\mathcal{R}(\{J_{i}\})$,
then
  \begin{equation}
  \triangle_{(K_{\lambda}\cup L_{\lambda})}(I_{i}\cup J_{i})
  =(K_{\lambda}\cup L_{\lambda})\otimes
  \textit{ind}\{(I_{i}\cup J_{i})\diagup(K_{\lambda}\cup L_{\lambda})\}.
  \end{equation}
\end{itemize}
\end{definition}

\begin{remark}
$\\$
\begin{itemize}
  \item Recalling the formulas (2.16) and corollary 2.2 we have
  $$
  \begin{array}{c}
  \triangle_{(K_{\lambda}\cup L_{\lambda})}(I_{i}\cup J_{i}) \\
  =(K_{\lambda}\cup L_{\lambda})\otimes
  (I_{i}\cup\mathcal{R}(J_{i})\diagup
  (id\times\mathcal{R}_{1})([I_{i}\cup J_{i}\diagup 
  (K_{\lambda}\cup L_{\lambda})]_{2})) \\
  =(M_{\mu})\diagup (\bigcup\limits_{i}J_{i})\otimes 
  (I_{i}\cup\mathcal{R}(J_{i}))\diagup(M_{\mu}),
  \end{array}
  $$
  where
  $$
  \{M_{\mu}\}=\bigcup\limits_{i}
  (id\times\mathcal{R}_{1})
  ([I_{i}\cup J_{i}\diagup(K_{\lambda}\cup L_{\lambda})]_{2})).
  $$
  Moreover, we have
  $$
  \bigcup\limits_{i}J_{i}\subset (M_{\mu})\subset
  (I_{i}\cup\mathcal{R}(J_{i})).
  $$
  \item Particularly, we have
  \begin{itemize}
  \item for $U\subset A$ and $(I_{i})\subset U$,
  if we identify $U$ with $U\cup\{\emptyset\}\in\Xi_{A}$,
  and $\{I_{i}\}$ with 
  $\{I_{i}\cup\emptyset\}\in\mathcal{P}_{dis}(\Xi_{A})$,
  we have
  \begin{equation}
  \triangle_{(I_{i})}=(I_{i})\otimes U\diagup (I_{i});
  \end{equation}
  \item for two partitions $(I_{i}),\,(J_{j})$ in $A$,
  $(J_{j})\subset(I_{i})$, we have
  \begin{equation}
  \triangle_{(J_{i})}(I_{i})=(\triangle_{(J_{i})}I_{i}).
  \end{equation}
  \end{itemize}
\end{itemize}
\end{remark}

Now we have the following lemma which is the
corollary of proposition 2.1 and 2.2.
\begin{lemma}
Let $U\cup\{I_{i}\}\in\Xi_{A}$, 
$\{D_{j}\cup E_{j}\},\,\{K_{\lambda}\cup L_{\lambda}\}\in
\mathcal{P}_{dis}(A)$, $(D_{j}\cup E_{j})\subset
(K_{\lambda}\cup L_{\lambda})\subset U\cup\{I_{i}\}$,
$\mathcal{R}(\{E_{j}\})=\mathcal{R}(\{L_{\lambda}\})=\{I_{i}\}$,
if we take 
$$
\{M_{\lambda}\cup N_{\lambda}\}=
\textit{ind}\{(K_{\lambda}\cup L_{\lambda})\diagup
(D_{j}\cup E_{j})\},
$$
we have
\begin{equation}
\begin{array}{c}
(\triangle_{(D_{j}\cup E_{j})}\otimes id)
\triangle_{(K_{\lambda}\cup L_{\lambda})}U\cup\{I_{i}\} \\
=(id\otimes\triangle_{(M_{\lambda}\cup N_{\lambda})})
\triangle_{(D_{j}\cup E_{j})}U\cup\{I_{i}\}.   
\end{array}
\end{equation}
\end{lemma}
\begin{proof}
The term $\triangle_{(M_{\lambda}\cup N_{\lambda})}
(\textit{ind}\{U\cup\{I_{i}\}\diagup(D_{j}\cup E_{j})\})$
will appear in (3.5), we need to test that
$(M_{\lambda}\cup N_{\lambda})$ satisfies the conditions
in definition 3.1. Recalling the discussions in subsection 2.6
and noting $\mathcal{R}(\{E_{j}\})
=\mathcal{R}(\{L_{\lambda}\})=\{I_{i}\}$ we know that 
$$
\textit{ind}\{U\cup\{I_{i}\}\diagup(D_{j}\cup E_{j})\}
=(U\setminus\mathcal{R}(\{D_{j}\}))\cup\{D_{j}
\cup\mathcal{R}(E_{j})\},
$$
and
$$
\textit{ind}\{K_{\lambda}\cup L_{\lambda}\diagup
(D_{j}\cup E_{j})\}=(K_{\lambda}\setminus\mathcal{R}(\{D_{j}\}))
\cup\{D^{(\lambda)}_{j}\cup\mathcal{R}(E^{(\lambda)}_{j})\},
$$
where $\{D^{(\lambda)}_{j}\cup\mathcal{R}(E^{(\lambda)}_{j})\}$
arises from the decomposition
$$
(D_{j}\cup E_{j})=\bigcup\limits_{\lambda}(D^{(\lambda)}_{j}\cup
E^{(\lambda)}_{j}),\,(D^{(\lambda)}_{j}\cup E^{(\lambda)}_{j})
\subset K_{\lambda}\cup L_{\lambda}.
$$
Thus the conditions in definition 3.1 are valid.
\end{proof}

\begin{definition}
We define the coproduct $\triangle_{1}$ as follows:
\begin{itemize}
  \item 
  \begin{equation}
  \triangle \emptyset=\emptyset\otimes\emptyset.
  \end{equation}
  \item Let $U\cup\{I_{i}\}\in\mathcal{P}_{dis}(\Xi_{A})$,
\begin{equation}
\begin{array}{c}
\triangle (U\cup\{I_{i}\})=U\cup\{I_{i}\}\otimes \emptyset+
\emptyset\otimes U\cup\{I_{i}\} \\
+\sum\limits_{(K_{\lambda}\cup L_{\lambda})\subset U\cup\{I_{i}\}}
\,\,\triangle_{(K_{\lambda}\cup L_{\lambda})}U\cup\{I_{i}\}.
\end{array}
\end{equation}
The sum in the formula (3.7) is over all partitions satisfying
the conditions in definition 3.1.
\end{itemize}
\end{definition}

\begin{remark}
$\\$
\begin{itemize}
\item
For the case of $\max\{\#[U],\,\#[\{I_{i}\}]\}\leqslant 2$, 
we have
$$
\triangle (U\cup\{I_{i}\})=U\cup\{I_{i}\}\otimes \emptyset+
\emptyset\otimes U\cup\{I_{i}\}.
$$
\item Combining definition 3.1 and 3.2 we have
\begin{equation}
\begin{array}{c}
\triangle(I_{i}\cup J_{i})=
(I_{i}\cup J_{i})\otimes\emptyset+\emptyset\otimes
(I_{i}\cup J_{i}) \\
+\sum\limits_{(K_{\lambda}\cup L_{\lambda})\subset(I_{i}\cup J_{i})}
\triangle_{(K_{\lambda}\cup L_{\lambda})}(I_{i}\cup J_{i}),
\end{array}
\end{equation}
where $\{I_{i}\cup J_{i}\}\in\mathcal{P}_{dis}(\Xi_{A})$
and $\mathcal{R}(\{L_{\lambda}\})=\mathcal{R}(\{J_{i}\})$.
\end{itemize} 
\end{remark}

About coproduct defined in definition 3.2 we have
\begin{theorem}
\begin{equation}
(\triangle\otimes id)\triangle 
=(id\otimes \triangle)\triangle.
\end{equation}
\end{theorem}

\begin{proof}
It is enough for us to consider the reduced
coproduct $\triangle^{\prime}$, where
$$
\triangle^{\prime}U\cup\{I_{i}\}
=\triangle U\cup\{I_{i}\}-(U\cup\{I_{i}\}\otimes\emptyset
+\emptyset\otimes U\cup\{I_{i}\}),
$$
and $U\cup\{I_{i}\}\in\mathcal{P}_{dis}(\Xi_{A})$.
According to the formulas (3.7) and (3.8) we have
$$
\begin{array}{c}
(\triangle^{\prime}\otimes id)\triangle_{1}^{\prime} U\cup\{I_{i}\} \\
=\sum\limits_{(K_{\lambda}\cup L_{\lambda})\subset U\cup\{I_{i}\}}
(\triangle^{\prime}\otimes id)
\triangle_{(K_{\lambda}\cup L_{\lambda})}(U\cup\{I_{i}\}) \\
=\sum\limits_{(K_{\lambda}\cup L_{\lambda})\subset U\cup\{I_{i}\}}\,\,
\sum\limits_{(D_{j}\cup E_{j})\subset(K_{\lambda}\cup L_{\lambda})}
(\triangle_{(D_{j}\cup E_{j})}\otimes id)
\triangle_{(K_{\lambda}\cup L_{\lambda})}(U\cup\{I_{i}\}),
\end{array}
$$
where $\mathcal{R}(\{E_{j}\})=\mathcal{R}(\{L_{\lambda}\})=\{I_{i}\}$

According to lemma 3.1 we have
$$
(\triangle_{(D_{j}\cup E_{j})}\otimes id)
\triangle_{(K_{\lambda}\cup L_{\lambda})}(U\cup\{I_{i}\})
=(id\otimes\triangle_{(M_{\mu}\cup N_{\mu})})
\triangle_{(D_{j}\cup E_{j})}(U\cup\{I_{i}\}),
$$
where 
$$
\{M_{\mu}\cup N_{\mu}\}=
\textit{ind}\{(K_{\lambda}\cup L_{\lambda})\diagup
(D_{j}\cup E_{j})\}.
$$
It is obvious that
$\{M_{\mu}\cup N_{\mu}\}\subset
\textit{ind}\{U\cup\{I_{i}\}\diagup(D_{j}\cup E_{j})\}$.
Conversely, by the procedure of reversion, it is easy to know
that for each 
$\{M_{\mu}\cup N_{\mu}\}\subset
\textit{ind}\{U\cup\{I_{i}\}\diagup(D_{j}\cup E_{j})\}$,
there is a $(K_{\lambda}\cup L_{\lambda})$ such that
$(D_{j}\cup E_{j})\subset(K_{\lambda}\cup L_{\lambda})$
and
$$
\{M_{\mu}\cup N_{\mu}\}=
\textit{ind}\{(K_{\lambda}\cup L_{\lambda})\diagup
(D_{j}\cup E_{j})\}.
$$
Up to now we complete the proof of the theorem.
\end{proof}

The following conclusion can be proved by induction.
\begin{proposition}
Let $m=\max\limits_{i}\{\#[U],\#[\{I_{i}\}]\}$, then we have
\begin{equation}
(\triangle^{\prime})^{m-1}(U\cup\{I_{i}\})=0.
\end{equation}
\end{proposition}

Now we consider the vector space over $\mathbb{C}$
spanned by $\mathcal{P}_{dis}(\Xi_{A})$ denoted by
$V_{\mathcal{P}_{A}}$. We define co-unit $\epsilon:V_{\mathcal{P}_{A}}\longrightarrow
\mathbb{C}$ as following:
$$
\epsilon(\emptyset)=1,\,\epsilon(U\cup\{I_{i}\})=0,
U\cup\{I_{i}\}\neq\emptyset.
$$
Then $(V_{\mathcal{P}_{A}},\triangle,\epsilon)$ is a coalgebra
if we extend $\triangle$ to $V_{\mathcal{P}_{A}}$. Actually, from
definition 2.6, it is obvious that we have
$$
\begin{array}{ccccc}
      & \triangle &  & \epsilon\otimes id &  \\
V_{\mathcal{P}_{A}} & \longrightarrow &V_{\mathcal{P}_{A}}
\otimes V_{\mathcal{P}_{A}}&
\longrightarrow &\mathbb{C}\otimes V_{\mathcal{P}_{A}}
\simeq  V_{\mathcal{P}_{A}},        
\end{array}
$$
and
$$
\begin{array}{ccccc}
      & \triangle &  & id\otimes \epsilon &  \\
V_{\mathcal{P}_{A}} & \longrightarrow &V_{\mathcal{P}_{A}}
\otimes V_{\mathcal{P}_{A}}&
\longrightarrow &V_{\mathcal{P}_{A}}\otimes \mathbb{C}
\simeq  V_{\mathcal{P}_{A}}.        
\end{array}
$$
From the coalgebra constructed above we can get
bialgebras $T(V_{\mathcal{P}_{A}})$ and
$S(V_{\mathcal{P}_{A}})$ in standard way, 
where $T(V_{\mathcal{P}_{A}})$
and $S(V_{\mathcal{P}_{A}})$ are tensor algebra and
symmetric tensor algebra of $V_{\mathcal{P}_{A}}$ respectively.
With the help of proposition 2.5 we know that
the reduced coproduct in definition 2.6 is conilpotent,
therefore $T(V_{\mathcal{P}_{A}})$ and
$S(V_{\mathcal{P}_{A}})$ are Hopf algebras.

\section{Hopf algebra concerning $gl(d,\mathbb{C})$}

In this section we will discuss Hopf algebra related to $gl(d,\mathbb{C})$
($d>1$). We will work on a subspace 
$$
gl(d,\mathbb{C})_{0}=\{M\in gl(d,\mathbb{C})|M\, with\, zero\, diagonal\}.
$$

\subsection{Quotient and Collapsing}

\paragraph{\textit{Diagonal submatrix:}}

Let $I\subset\underline{d}=\{1,\cdots,d\},
\,|I|=k,\, I=\{a_{1},\cdots, a_{k}\}\,(2\leqslant k<d;\,
0<a_{1}<\cdots <a_{k}\leqslant d)$, 
$M=(m_{ij})_{d\times d}\in gl(d,\mathbb{C})_{0}$,
then $I$ determines a diagonal submatrix of $M$ denoted by $M_{I}$,
$M_{I}=(m_{a_{i}a_{j}})_{k\times k}$. 
The subset $I$ is called the position
of $M_{I}$. Conversely, we can define the embedding 
$\iota_{I}:gl(k,\mathbb{C})_{0}\hookrightarrow gl(d,\mathbb{C})_{0}$,
for $M=(m_{ij})_{k\times k}\in gl(k,\mathbb{C})_{0}$, 
$\iota_{I}M=(m_{ij}^{\prime})_{d\times d}$, such that 
$m^{\prime}_{a_{i}a_{j}}=m_{ij}$, $m^{\prime}_{pq}=0$ 
($p\in I^{c}$ or $q\in I^{c}$). The subset $I$ 
is called the position of $\iota_{I}$. Actually, we have
$(\iota_{I}M_{I})_{I}=M_{I}$, in this sense we can 
identify $M_{I}$ with $\iota_{I}M_{I}$.

For two subsets $I,J\subset \underline{d}$, it is easy to check that
$$
\iota_{I}(\iota_{J}M_{J})_{I}=\iota_{J}(\iota_{I}M_{I})_{J}=\iota_{I\cap J}M_{I\cap J},
$$
specially, if $I\subset J$, we have
$$
\iota_{I}(\iota_{J}M_{J})_{I}=\iota_{I}M_{I}.
$$
We can always think of $\emptyset\subset\underline{d}$ and $\emptyset$
as a matrix of order $0$. It is natural for us to define
$$
\iota_{\emptyset}\emptyset=0\in gl(d,\mathbb{C})_{0},\,\forall\,d\in \mathbb{N}.
$$

Generally, for a partition $(I_{1},\dots,I_{l})$ in $\underline{d}$, $M\in gl(d,\mathbb{C})_{0}$,
we define
$$
M_{(I_{i})}=\sum\limits_{i=1}^{l}\iota_{I_{i}}M_{I_{i}}.
$$
For two partitions $(I_{1},\dots,I_{l})$ and $(J_{1},\dots,J_{k})$, we have
$$
M_{(I_{i})\cap (J_{j})}=\sum\limits_{i,j}\iota_{I_{i}\cap J_{j}}M_{I_{i}\cap J_{j}}.
$$
If $(I_{i})\cap (J_{j})=\emptyset$, we have
$$
M_{(I_{i})\cup (J_{j})}=\sum\limits_{i}\iota_{I_{i}}M_{I_{i}}+\sum\limits_{j}\iota_{J_{j}}M_{J_{j}}.
$$

\paragraph{\textit{Quotient and collapsing:}}

Let $M\in gl(d,\mathbb{C})_{0}$, $I\subset \underline{d}$,
$I=\{a_{1},\cdots, a_{k}\}$, $I^{c}=\{b_{1},\dots,b_{p}\}$
($I^{c}=\underline{d}\setminus I;\,0<a_{1}<\dots <a_{k}\leqslant d; 0<b_{1}<\dots <b_{p}\leqslant d;
\,1<k<d,\,p=d-k$),
we define collapsing matrix of $M$ denoted by $M\diagup M_{I}$
in the following way. To get $M\diagup M_{I}$
we extend $M_{I^{c}}$ by putting an "ideal index"
$\ast$ of row and column,
$$
M\diagup M_{I}=
\begin{pmatrix}
      0  & m_{1^{\ast} 2}\cdots m_{1^{\ast} p+1}    \\
  \begin{array}{c}
      m_{2 1^{\ast}}    \\
      \vdots    \\
      m_{p+1 1^{\ast}}   
\end{array}      
      &  M_{I^{c}}
\end{pmatrix},
$$
where 
$m_{1^{\ast} i}=\sum_{j=1}^{k}m_{a_{j}b_{i-1}},m_{i1^{\ast}}=\sum_{j=1}^{k}m_{b_{i-1}a_{j}}
(\,i=2,\dots, p+1)$. From the definition of collapsing mentioned above we 
know that $M\diagup M_{I}\in gl(d-k+1,\mathbb{C})_{0}$.
The set of indices of rows or columns of $M\diagup M_{I}$
consists of $I^{c}$ and set of "ideal index" $\{\ast\}$.
$M\diagup M_{I}$ is called quotient of $M$ by $M_{I}$.
Particularly, we define 
\begin{equation}
M\diagup M:=0, \,M\diagup 0:=M.
\end{equation}

\paragraph{Example:}
In this example we calculate the quotient of quotient.
For a non-trivial subset $I\subset \underline{d}$,
the set of indices of rows or columns of quotient 
$M\diagup M_{I}$ consists
of $I^{c}\cup \{\ast\}$, where $M\in gl(d,\mathbb{C})_{0}$.
Let $J\subset I^{c}\cup \{\ast\}$ be a non-trivial subset, 
we can define $(M\diagup M_{I})\diagup (M\diagup M_{I})_{J}$
in the way same as $M\diagup M_{I}$.
Let $J^{\prime}=J\cap I^{c}$, we consider two possible
cases:
\begin{itemize}
\item \textbf{Case of $J\subset I^{c}$:}
In this situation we have $I,J\subset \underline{d}$, 
$I\cap J=\emptyset$ and $(M\diagup M_{I})_{J}=M_{J}$.
To get $(M\diagup M_{I})\diagup (M\diagup M_{I})_{J}$
we need to put an additional "ideal index" of row and column.
Explicitly, let $\underline{d}\setminus (I\cup J)=\{i_{1},\cdots,i_{q}\}$
($0<i_{1}<\cdots<i_{q}\leqslant d$),
then the set of indices of rows and columns for
$(M\diagup M_{I})\diagup (M\diagup M_{I})_{J}$ is
$\{1^{\ast},2^{\ast},1,\cdots,q\}$. If in 
$(M\diagup M_{I})\diagup (M\diagup M_{I})_{J}$
we let $1^{\ast}$ corresponds to $I$, and $2^{\ast}$
corresponds to $J$, 
then $m_{1^{\ast}2^{\ast}}$ is the sum of the entries with indices of
rows in $I$ and indies of columns in $J$.
$m_{1^{\ast}j}$ is the sum of entries with index of column $i_{j}$
and indices of rows in $I$ ($1\leqslant j\leqslant q$).
The other entries with "ideal index" are similar.
Moreover, it is easy to check that
$$
(M\diagup M_{I})\diagup (M\diagup M_{I})_{J}
=(M\diagup M_{J})\diagup (M\diagup M_{J})_{I}.
$$
\item \textbf{Case of $J=J^{\prime}\cup \{\ast\}$:}
It is easy to check that
$$
(M\diagup M_{I})\diagup (M\diagup M_{I})_{J}=M\diagup M_{(I\cup J^{\prime})},
$$
and
$$
(M\diagup M_{I})_{J}=(M_{I\cup J^{\prime}})\diagup M_{I}.
$$
\end{itemize}

\begin{remark}
$\\$
\begin{itemize}
  \item Actually, for a given subset $I\subset\underline{d}$,
  we have a "factorization" of $M$ according to $M_{I}$:
  
  $$
  M=M_{1}+\iota_{I}M_{I}+\iota_{I^{c}}M_{I^{c}}.
  $$
  The procedure of collaping takes place on 
  $M_{1}=M-\iota_{I}M_{I}-\iota_{I^{c}}M_{I^{c}}$.
  The quotient can be described in the following way.
  
  \begin{itemize}
  \item $\mathbf{Step\,\,1:}$ $M_{I}$ is removed from $M$. 
  \item $\mathbf{Step\,\,2:}$ $M_{1}$ is collapsed into
  a matrix $M^{\ast}$ ,where $M^{\ast}\in gl(d-k+1,\mathbb{C})_{0}$
  with form as follows:
  
  $$
  M^{\ast}=
  \begin{pmatrix}
  0  & m_{1^{\ast} 2}\cdots m_{1^{\ast} p+1}    \\
  \begin{array}{c}
  m_{2 1^{\ast}}    \\
  \vdots    \\
  m_{p+1 1^{\ast}}   
  \end{array}      
  &  0
  \end{pmatrix}.
  $$
  
  \item $\mathbf{Step\,\,3:}$ $M_{I^{c}}$ is embedded into
  $gl(d-k+1,\mathbb{C})_{0}$ where the position is
  $\{2,\cdots,m-k+1\}$. The embedding is also denoted by $\iota_{I^{c}}$.
  \end{itemize}
  Finally, we get the quotient of $M$:
  
  $$
  M\diagup M_{I}=M^{\ast}+\iota_{I^{c}}M_{I^{c}}.
  $$
  As a consequence of the above discussion, we have
  
  $$
  (M_{1}+M_{2})\diagup (M_{1}+M_{2})_{I}=
  M_{1}\diagup(M_{1})_{I}+M_{2}\diagup(M_{2})_{I},
  M_{1},M_{2}\in gl(d,\mathbb{C})_{0}.
  $$ 
  Thus, for a given subset $I\subset\underline{d}$
  ($|I|=k,\,2\leq k< d$), the quotient 
  
  $$
  \cdot\diagup I:gl(d,\mathbb{C})_{0}
  \longrightarrow gl(d-k+1,\mathbb{C})_{0}
  $$
  is a homomorphism, where $gl(d,\mathbb{C})_{0}$
  and $gl(d-k+1,\mathbb{C})_{0}$ are regarded
  as Able groups under the addition of the
  matrices. 
  
\item If $I=\{i\}$ and we put "ideal index" in the original position labeled 
  $i$, then we have $M=M\diagup M_{I}$.
\end{itemize}
\end{remark}

Generally, we can discuss the case of partitions.
For a partition $(I_{1},\dots,I_{l})$ in $\underline{d}$, $\sum_{i=1}^{l}|I_{i}|<d$,
$1<|I_{i}|\,(1=1,\dots, l)$, we can define the quotient of $M$ by $M_{(I_{i})}$ 
denoted by $M\diagup M_{(I_{i})}$ (or by $M\diagup (I_{i})$ simply) inductively.

$$
M\diagup (I_{i})=(\dots((M\diagup M_{I_{1}})\diagup M_{I_{2}}\dots)\diagup M_{I_{l}}.
$$

\begin{remark}
$\\$
\begin{itemize}
\item We specify "ideal indices" of rows or columns of $M\diagup(I_{i})$
situate in up $l$ rows and left $l$ columns. 
\item Let $\sigma\in \textbf{S}_{l}$, $\textbf{S}_{l}$ denotes the symmetric
group of $l$ letters, if we ignore the order of "ideal indices" of
$M\diagup(I_{i})$, we do not distinguish $M\diagup (I_{i})$
from $M\diagup (I_{\sigma(i)})$.
\end{itemize}
\end{remark}

Here we are interested in the case of $(M\diagup (I_{i}))\diagup (J_{j})$,
where $(J_{1},\cdots,J_{k})$ is a partition in
$I^{c}\cup \{1^{\ast},\cdots,l^{\ast}\}$, $i^{\ast}\,\,(1\leq i\leq l)$ is 
"ideal index" of $M\diagup (I_{i})$ corresponding to $I_{i}$ and 
$I=\mathcal{R}(\{I_{i}\}),\,I^{c}=\underline{d}\setminus I$. We hope to
express $(M\diagup (I_{i}))\diagup (J_{j})$ in terms of $M$ and
partitions in $\underline{d}$. 
Actually, recalling the contents in section 2,
the set of indices of rows or columns of $M\diagup (I_{i})$,
$I^{c}\cup \{1^{\ast},\cdots,l^{\ast}\}$, can be identified
with $\underline{d}\diagup (I_{i})=I^{c}\cup \{I_{i}\}$,
where we identify $i^{\ast}$ with $I_{i}$.
Then we have:

\begin{proposition}
Let $(I_{i})$ and $(K_{\lambda})$ be partitions in $\underline{d}$,
$(I_{i})\subset (K_{\lambda})$, $(J_{j})=(K_{\lambda})\diagup (I_{i})$,
then we have

\begin{equation}
(M\diagup (I_{i}))\diagup (J_{j})=M\diagup (K_{\lambda}),
\end{equation}
and

\begin{equation}
(M\diagup (I_{i}))_{(J_{j})}=M_{(K_{\lambda})}\diagup (I_{i}),
\end{equation}
where in the formula (4.3) we have 
$\mathcal{R}_{1}(J_{j})\subset K_{\lambda}$ and
$I_{i}\subset K_{\lambda}$.

Conversely, for a partition $(I_{1},\cdots,I_{l})$ in $\underline{d}$ and
partition $(J_{j})$ in $\underline{d}\diagup(I_{i})=
(\underline{d}\setminus\mathcal{R}(\{I_{i}\}))\cup\{1^{\ast},\cdots,l^{\ast}\}$,
there is a partition $(K_{\lambda})$ in $\underline{d}$, such that
 
$$
\begin{array}{c}
(I_{i})\subset (K_{\lambda}),\,
(K_{\lambda})=(\mathcal{R}_{1}(J_{j}))\cup
(I_{i})_{i^{\ast}\notin\mathcal{R}_{1}(\{J_{j}\})}, \\
(K_{\lambda})\diagup (I_{i})=
(\emptyset\cup\{I_{i}\})_{i^{\ast}\notin
\mathcal{R}_{1}(\{J_{j}\})}
\cup(J_{j}),
\end{array}
$$
and the formulas (4.2), (4.3) are valid.
\end{proposition}

\begin{proof}
The first part of the proposition is obvious, we need to prove
the second part.
Let $J=\bigcup_{j=1}^{k}J_{j}$, $J^{\prime}=J\cap I^{c}$, 
$J^{\prime\prime}=J\cap \underline{l}^{\ast}$, 
$J^{\prime}_{j}=J_{j}\cap I^{c}$, $J_{j}^{\prime\prime}=J_{j}\cap \underline{l}^{\ast}$
($j=1,\cdots,k$), where $\underline{l}^{\ast}=\{1^{\ast},\cdots,l^{\ast}\}$ is the set of
"ideal indices" of $M\diagup (I_{i})$. We discuss the problem for three
cases respectively.

\textbf{Case of $J^{\prime\prime}=\emptyset$:}
In this case $J\subset I^{c}$, thus $(J_{j})$ is a partition in $I^{c}$.
Moreover, we know that $(M\diagup (I_{i}))_{(J_{j})}=M_{(J_{j})}$.
If we take $(K_{\lambda})=(I_{i})\cup(J_{j})$, it is obvious that 
the formulas (4.2), 4.3) are valid.

\textbf{Case of $J^{\prime\prime}=\underline{l}^{\ast}$:}
In this case we take 
$K_{j}=J_{j}^{\prime}\cup(\bigcup_{i^{\ast}\in J_{j}^{\prime\prime}}I_{i})$
for $J_{j}^{\prime\prime}\neq \emptyset$ and
$K_{j}=J_{j}$ for $J_{j}^{\prime\prime}=\emptyset$.
Then $(I_{i})\subset (K_{\lambda})$ and (4.2), (4.3) are valid.

\textbf{Csae of $J^{\prime\prime}\neq\emptyset$ and $\underline{l}^{\ast}
\setminus J^{\prime\prime}\neq\emptyset$:}
Without loss of generality, we assume $J^{\prime\prime}\neq\emptyset$
($1\leqslant j\leqslant p$), $J^{\prime\prime}=\emptyset$
($j>p$). We take $(K_{\lambda})$ in the following way:
$$
\Bigg\{
\begin{array}{ccc}
      K_{\lambda}=&J_{\lambda}^{\prime}\cup
      (\bigcup\limits_{i^{\ast}\in J_{\lambda}^{\prime\prime}}I_{i}),
      & 1\leqslant\lambda\leqslant p   \\
      K_{\lambda}=&J_{\lambda},  & p<\lambda\leqslant k \\
      K_{\lambda}=&I_{i_{\lambda}}, & i_{\lambda^{\ast}}\in 
      \underline{l}^{\ast}\setminus J^{\prime\prime}.  
\end{array}
$$
It is obvious that $(I_{i})\subset (K_{\lambda})$, and we can
check that (4.2), (4.3) are valid. 
\end{proof}

\subsection{Hopf algebra related to $gl(d,\mathbb{C})_{0}$}

\paragraph{\textit{Coproduct:}}
Let $(I_{i})$ be a partition in $\underline{d}$ we define
"sub-coproduct" related to $(I_{i})$ as follows:

\begin{definition}Let $M\in gl(d,\mathbb{C})_{0}$
($d>2$), $(I_{i})$ be a partition in $\underline{d}$,
$|I_{i}|>1$ for each $i$, we define

\begin{equation}
\triangle_{(I_{i})}M=M_{(I_{i})}\otimes M\diagup (I_{i}).
\end{equation}
\end{definition}

The following lemma is a corollary of proposition 2.1.

\begin{lemma}
Let
$(I_{i})$, $(K_{\lambda})$ be two partitions in $\underline{d}$
satisfying $(I_{i})\subset (K_{\lambda})$, then we have
$$
(\triangle_{(I_{i})}\otimes id)\triangle_{(K_{\lambda})}
=(id\otimes\triangle_{(K_{\lambda})\diagup (I_{i})})\triangle_{(I_{i})}.
$$
\end{lemma}

For coproduct we have the following definition:

\begin{definition}
We define the coproduct as follows:
\begin{itemize}
\item 

\begin{equation}
\triangle\emptyset=\emptyset\otimes\emptyset.
\end{equation}
\item Let $M\in gl(2,\mathbb{C})_{0}$,

\begin{equation}
\triangle M=M\otimes\emptyset+\emptyset\otimes M.
\end{equation}
\item Let $M\in gl(d,\mathbb{C})_{0}$, $d>2$,

\begin{equation}
\triangle M=M\otimes \emptyset+\emptyset\otimes M
+\sum\limits_{I\subset \underline{d},\,(I_{i})\in \textbf{part}(I)}
\triangle_{(I_{i})}M
\end{equation}
In sum (4.7), for all partitions $(I_{i})$ we assume $|I_{i}|>1$,
moreover, if $I=\underline{d}$, $(I_{i})$ consisting of
at least two subsets.
\end{itemize}
\end{definition}

The coproduct defined in definition 4.2 is coassociative.
Actually we have

\begin{theorem}
The coproduct in definition 4.2 satisfies

\begin{equation}
(\triangle\otimes id)\triangle=(id\otimes\triangle)\triangle.
\end{equation}
Furthermore, for $M\in gl(d,\mathbb{C})_{0}$
($M\not=0$), we have

\begin{equation}
(\triangle^{\prime})^{d-1}M=0,
\end{equation}
where $\triangle^{\prime}$ is reduced coproduct

$$
\triangle^{\prime}M=\triangle M-(M\otimes\emptyset+\emptyset\otimes M).
$$
\end{theorem}

\begin{proof}
We need only to check the formula for reduced coproduct.
For $(\triangle^{\prime}\otimes id)\triangle^{\prime}$ we have

$$
\begin{array}{c}
      (\triangle^{\prime}\otimes id)\triangle^{\prime}M    \\
     =\sum\limits_{(I_{i})\in \textbf{part}(I),\,I\subset \underline{d}}
     \triangle^{\prime} M_{(I_{i})}\otimes M\diagup (I_{i}) \\
     =\sum\limits_{(I_{i})\in \textbf{part}(I),\,I\subset \underline{d}}\,\,
     \sum\limits_{(J_{j})\subset (I_{i})}(M_{(I_{i})})_{J_{j}}\otimes
     M_{(I_{i})}\diagup (J_{j})\otimes M\diagup (I_{i}).  
\end{array}
$$
Because $(J_{j})\subset (I_{i})$, we have 
$(M_{(I_{i})})_{(J_{j})}=M_{(J_{j})}$.

On the other hand, for $(id\otimes\triangle^{\prime})\triangle^{\prime}$
we have

$$
\begin{array}{c}
       (id\otimes\triangle^{\prime})\triangle^{\prime}M   \\
       =\sum\limits_{(J_{j})\in \textbf{part}(J),\,J\subset \underline{d}}
       M_{(J_{j})}\otimes \triangle^{\prime}(M\diagup (J_{j})) \\
       =\sum\limits_{(J_{j})\in \textbf{part}(J),\,J\subset \underline{d}}\,\,
       \sum\limits_{(K_{\lambda})\in \textbf{part}(K),\,K\subset J^{c}\cup
       \{\ast,\cdots,\ast\}}M_{(J_{j})}\otimes (M\diagup (J_{j}))_{(K_{\lambda})} 
       \otimes(M\diagup (J_{j}))\diagup (K_{\lambda}).
\end{array}
$$
From proposition 4.1 we know that there is a partition
$(I_{i})$ in $\underline{d}$ such that $(J_{j})\subset (I_{i})$
and

$$
\begin{array}{c}
(K_{\lambda})=(I_{i})\diagup (J_{j}), \\
(M\diagup (J_{j}))\diagup (K_{\lambda})=M\diagup (I_{i}), \\
(M\diagup (J_{j}))_{(K_{\lambda})}=M_{(I_{i})}\diagup (J_{j}).
\end{array}
$$
Comparing the expressions of both of 
$(\triangle^{\prime}\otimes id)\triangle^{\prime}M$ and
$ (id\otimes\triangle^{\prime})\triangle^{\prime}M$, 
we know that the formula (4.8) is valid. Noting the formula
(4.6) in definition 4.2, The formula (4.9) can be proved
by induction obviously.
 
\end{proof}

Let 
$$
C_{d}=\bigoplus\limits_{0\leqslant k\leqslant d}gl(k,\mathbb{C})_{0},
$$
and we define unit $u:\mathbb{C}\to C_{d}$ and counit 
$\eta:C_{d}\to\mathbb{C}$ of $C_{d}$ 
as follows:

\begin{equation}
u:c\mapsto c\,0,
\end{equation} 

\begin{equation}
\eta:0\mapsto 1,\eta:M\mapsto 0\,(M\neq 0),
\end{equation}
then $C_{d}$ is a coalgebra. Furthermore,
$T(C_{d})$ and $S(C_{d})$ are Hopf algebras.

\paragraph{\textit{Hpof algebra $\mathcal{H}_{gl}$:}}

Let $A,B\in gl(d,\mathbb{C})_{0}$, we define
a equivalent relation as follows:

$$
A\thicksim B\,\Longleftrightarrow
\exists\, P\,\,s.t.\,\,A=PBP^{T},
$$
where $P$ is a permutation matrix.
In other word, let $A=(a_{ij})_{d\times d}$,
$B=(b_{ij})_{d\times d}$, then $A\thicksim B$ if and only if
there is a $\pi\in \textbf{S}_{d}$ such that $a_{ij}=b_{\pi(i),\pi(j)}$.
For $M=(m_{ij})_{d\times d}\in gl(d,\mathbb{C})_{0}$, we set 
$\pi(M)=(m_{\pi(i),\pi(j)})$, where
$\pi\in \textbf{S}_{d}$.
The equivalent class of a matrix $M$ is denoted by
$\{M\}=\{\pi(M)|\pi\in \textbf{S}_{d}\}$
and the set of equivalent class in $gl(d,\mathbb{C})_{0}$ is denoted
by $(gl(d,\mathbb{C})_{0})\diagup\sim$. Let $c\in\mathbb{C}$, 
$\{M\}\in (gl(d,\mathbb{C})_{0})\diagup\sim$,
we define

$$
c\{M\}=\{cM\},
$$
then $(gl(d,\mathbb{C})_{0})\diagup\sim$ is a vector
space over $\mathbb{C}$.

Let 
$I=\{i_{1},\cdots,i_{k}\}\subset \underline{d}$, $\pi\in \textbf{S}_{d}$,
$\pi(I)=\{\pi(i_{1}),\cdots,\pi(i_{k})\}$, it is easy to check that

$$
M_{I}\sim \pi(M)_{\pi^{-1}(I)},\,
M\diagup M_{I}\thicksim \pi(M)\diagup \pi(M)_{\pi^{-1}(I)}.
$$
Above facts are valid for the case of partitions obviously.
For sub-coproduct (4.4) we define the action of $\pi\in \textbf{S}_{d}$
in the following way:

\begin{equation}
\pi(\triangle_{(I_{i})}M)=\triangle_{(\pi^{-1}(I_{i}))}\pi(M)
=\pi(M)_{(\pi^{-1}(I_{i}))}\otimes \pi(M)\diagup (\pi^{-1}(I_{i})).
\end{equation}
Now we can define the action of $\pi\in \textbf{S}_{d}$
on coproduct, here we discuss (4.7) only, in the following way

\begin{equation}
\pi(\triangle M)=\pi(M)\otimes\emptyset+\emptyset\otimes\pi(M)
+\sum\pi(\triangle_{(I_{i})}M).
\end{equation}
With the help of (4.13) we can extend the coproduct to the
case of equivalent class naturally.

\begin{equation}
\triangle\{M\}=\{M\}\otimes \emptyset+\emptyset\otimes\{M\}+
\sum\{M_{(I_{i})}\}\otimes\{M\diagup (I_{i})\}.
\end{equation}

Now we discuss the multiplication for equivalent classes
mentioned above with the help of the direct sum of the
matrices.

\begin{definition}
The multiplication $\odot$ is a map

$$
\odot:(gl(k,\mathbb{C})_{0})_{\thicksim}\times
(gl(l,\mathbb{C})_{0})_{\thicksim}\longrightarrow
(gl(k+l,\mathbb{C})_{0})_{\thicksim},
$$
\begin{equation}
\{M\}\odot\{N\}=\{\mathbf{diag}(M,N)\},
\,M\in gl(k,\mathbb{C})_{0},
N\in gl(l,\mathbb{C})_{0}.
\end{equation}
\end{definition}

Because of

$$
\{\mathbf{diag}(\pi_{1}(M),\pi_{2}(N))|
\pi_{1}\in\mathbb{S}_{k},\pi_{2}\in\mathbb{S}_{l}\}
\subset\{\pi(\mathbf{diag}(M,N))
|\,\pi\in\mathbb{S}_{k+l}\},
$$
the multiplication (4.15) is well defined.
The multiplication $\odot$ is commutative obviously.
Moreover, let $M_{i}\in gl(k_{i},\mathbb{C})_{0}$
($i=1,2,3$), similar to the previous discuussion
we can see that

$$
(\{M_{1}\}\odot\{M_{2}\})\odot\{M_{3}\}=
\{M_{1}\}\odot(\{M_{2}\}\odot\{M_{3}\}).
$$
The multiplication $\odot$ can be extended to
the situation of tensor. Let 
$M_{i}\in gl(k_{i},\mathbb{C})_{0}$,
$N_{i}\in gl(l_{i},\mathbb{C})_{0}$
($i=1,\cdots,n$), we now define the multiplication
of the tensors to be

$$
\begin{array}{c}
(\{M_{1}\}\otimes\cdots\otimes\{M_{n}\})\odot
(\{N_{1}\}\otimes\cdots\otimes\{N_{n}\}) \\
=(\{M_{1}\}\odot\{N_{1}\})\otimes\cdots\otimes
(\{M_{n}\}\odot\{N_{n}\}).
\end{array}
$$
Let

\begin{equation}
\mathcal{H}_{gl}=\mathbf{Span}_{\mathbb{C}}
(\{\odot_{i=1}^{k}\{M_{i}\}
|M_{i}\in gl(d_{i},\mathbb{C})_{0},
d_{i},k\in\mathbb{N},1\leq i\leq k\}),
\end{equation}
then the coproduct can be naturally extended
to $\mathcal{H}_{gl}$. For
$M_{i}\in gl(k_{i},\mathbb{C})_{0}$ ($i=1,2$),
we define

$$
\triangle(\{M_{1}\}\odot\{M_{2}\})=
\triangle\{M_{1}\}\odot\triangle\{M_{2}\}.
$$

The unit $u$ and counit $\eta$ on $\mathcal{H}_{gl}$ 
are defined as follows:

\begin{equation}
u:c\mapsto 0
\end{equation}
 
\begin{equation}
\eta:0\mapsto 1,\eta:\{M\}\mapsto 0,
\{M\}\not=\{0\}.
\end{equation}
It is obvious that $\mathcal{H}_{gl}$ is a  bialgebra.
By theorem 4.1 we know that $\mathcal{H}_{gl}$
is conilpotent, thus, it is a Hopf algebra.

\section{Star product}

\subsection{Notations}

Following the idea in \cite{8} we construct the star product
of scalar fields starting from a specific class of Kontsevich's graphs,
called the Bernoulli graphs.
At first we recall some notations about Kontsevich's graphs.

\begin{definition}
(\textbf{Admissible graphs}, V.Kontsevich \cite{7}p.22) 
Admissible graph $G_{n,m}$ is an oriented graph with labels such that
\begin{itemize}
  \item The set of vertices $V_{\Gamma}$ is $\{1,\cdots,n\}\sqcup\{\bar{1},\cdots,\bar{m}\}$ 
  where $n,m\in \mathbb{Z}_{\geqslant 0},\, 2m+n-2\geqslant 0$; vertices from $\{1,\cdots,n\}$
  are called vertices of the first type, vertices from $\{\bar{1},\cdots,\bar{m}\}$ are called vertices
  of the second type.
  \item Every edge $e=(v_{1},v_{2})\in E_{\Gamma}$ stars at a vertex of the first type,
  $v_{1}\in \{1,\cdots,n\}$.
  \item There are no loops, i.e. no edges of the type $(v,v)$.
  \item For every vertex $k\in\{1,\cdots,n\}$ of the first type, the set of edges
  $$Star(k)=\{(v_{1},v_{2})\in E_{\Gamma}|v_{1}=k\}$$
  starting from $k$, is labeled by symbols $\{e_{k}^{1},\cdots,e_{k}^{\#Star(k)}\}$.
\end{itemize}
\end{definition}

\begin{definition}
 (see  L. M. Ionescu\cite{5} and V. Kathotia\cite{6})
 If $\Gamma_{1}\in G_{n,m}$, $\Gamma_{2}\in G_{n^{\prime},m}$,
 we define the product $\Gamma_{1}\Gamma_{2}\in G_{n+n^{\prime},\,m}$ as the graph obtained
 from disjoint union of two graphs by identification of the vertices of the second type. 
\end{definition}

\begin{definition}
An adjacency matrix
is a symmetric matrix with non-negative integer entries
and zeros along the main diagonal. 
We call $\sum_{ij}m_{ij}$ the
degree of $M$ denoted by $degM$.  The set of adjacency matrices of
$d\times d$ is denoted by $M_{adj}(d,\mathbb{N})$.
\end{definition}

\begin{definition}
For a $m\times m$ adjacency matrix $M$ with $degM=k$, a Bernoulli graph corresponding
to $M$ is $b_{M}=\prod_{i<j}b^{m_{ij}}_{ij}\in G_{k,m}$, where $b_{ij}=\iota_{ij}b_{1}$,
$b_{1}\in G_{1,2}$ is a 
Kontsevich graph with one vertex of the first type endowed with two edges ending at two
vertices of the second type respectively, and
$\iota_{ij}:G_{1,2}\to G_{1,m}$ is an embedding with position $\{i,j\}\,(i<j)$.
\end{definition}

\begin{remark}
$\\$
\begin{itemize}
  \item In definition 4.4 the embedding $\iota_{ij}$ was introduced in ZhouMai\cite{8}. The
basic Bernoulli graph $b_{1}$ is referred to
L.M.Ionescu\cite{5} and V.Kathotia\cite{6}.
Because $b_{ij}$ represents a graph with $m$ vertices of the second type, one vertex of the
first type and two edges starting from unique vertex of the first type, we can think $b_{ij}$
is assigned to this vertex of the first type and two edges ending at $i-$th and $j-$th
vertices of the second type respectively. If $M=0$, $b_{M}=\emptyset$.

\item The expression $b_{M}=\prod_{i<j}b^{m_{ij}}_{ij}$ means that we do not distinguish
any two vertices of the first type connect with same two vertices of the second type. 
That can not lead to confusion (see Zhoumai\cite{8}).
\item $\{b_{ij}\}_{1\leqslant i<j\leqslant m}$ generats a monoid
  
  $$
  B_{m}=\{b_{M}|M\in M_{adj}(m,\mathbb{N})\},
  $$
moreover, generats the free algebra over $C$, $\mathbf{Span}_{\mathbb{C}}(B_{m})$
  (see ZhouMai\cite{8}). We call $b_{ij}$ the basic Bernoulli graph.
  \item For an adjacency matrix $M$, graph $b_{M}$ corresponds to a Feynman diagram
  (see ZhouMai\cite{8}).
\end{itemize}
\end{remark}

\subsection{Star product of scalar fields}

In this subsection we briefly recall the contents of \cite{8} 
(the datails refer to \cite{8}).
Firstly we discuss star product at level of functions not composing with fields.
We recall Kontsevich's rule,
here we modify Kontsevich's rule slightly, the poly-vector fields and poly-differential operators
are taken to be tensor forms instead of ordinary ones. Let $\mathcal{A}$
be an algebra generated by $\{K_{ij}|i,j\in \mathbb{Z}^{+}\}$, here $K_{ij}$
are abstract elements playing the role of coefficients of Poisson bi-vector
field which is 

$$
\mathcal{K}=\sum\limits_{i<j}\mathcal{K}_{ij},
$$
where $\mathcal{K}_{ij}=K_{ij}\partial z_{i}\otimes\partial z_{j}$.
Recalling the contents about Kontsevich's rule in
\cite{8}, now we have:

\textbf{Kontsevich's rule:}
\begin{itemize}
  \item $i-$th ($1\leqslant i\leqslant m$) vertex of the second type is assigned
  to a smooth function $f_{i}(z_{i})\in\mathbf{C}^{\infty}(\mathbb{R})$;
  \item For a basic Bernoulli graph $b_{ij}$, two edges starting at the unique vertex of the first type in
  $b_{ij}$ are assigned to $\partial z_{i}$ and $\partial z_{j}$
  according to that the end point is $i-$th or $j-$th vertex of the second type.
  The unique vertex of first type is assigned to "coefficient" $K_{ij}$.
  Thus $b_{ij}$ is assigned to a bi-differential operator
    $K_{ij}\partial z_{i}\otimes\partial z_{j}$
  denoted by
  
\begin{equation}
\mathcal{U}(b_{ij},\mathcal{K})=\mathcal{K}_{ij}.
\end{equation}
  \item For the general Bernoulli graphs, for example, 
  $b_{i_{1}j_{1}}\cdots b_{i_{k}j_{k}}$, it is
  assigned to a poly-differential operator
  $$
  \mathcal{U}(b_{i_{1}j_{1}}\cdots b_{i_{k}j_{k}},\mathcal{K})=
  \mathcal{K}_{i_{1}j_{1}}\cdots\mathcal{K}_{i_{k}j_{k}}.
  $$
\end{itemize}

With the help of the notation of adjacency matrices. we know that

$$
\mathcal{U}(b_{M},\mathcal{K})=\mathcal{K}_{M}
=K_{M}\partial_{z_{1}}^{\alpha_{1}}\cdots
\partial_{z_{m}}^{\alpha_{m}},\,\,
M\in M_{adj}(m,\mathbb{N}),
$$
where $\mathcal{K}_{M}=\sum_{i<j}\mathcal{K}_{ij}^{m_{ij}}$,
$K_{M}=\sum_{i<j}K_{ij}^{m_{ij}}$, and

$$
\alpha_{i}=\sum\limits_{j}m_{ij},\,i=1,\cdots,m.
$$
Furthermore, we have

$$
\mathcal{U}(b_{M_{1}}b_{M_{2}},\mathcal{K})=
\mathcal{U}(b_{M_{1}},\mathcal{K})
\mathcal{U}(b_{M_{2}},\mathcal{K}).
$$
Therefore, we get a homomorphism:

$$
\mathcal{U}(\cdot,\mathcal{K}):
\mathbf{Span}_{\mathbb{C}}(B_{m})\longrightarrow
\left\{
\begin{array}{c}
set\,\,of\,\,the\,\,poly-differential  \\
operators\,\,with\,\,coefficients\,\,in\,\,\mathcal{A}
\end{array}
\right\}
$$

Due to Kontsevich's rule, with some slight modification here,
the star product can be expressed by means of Bernoulli graphs as following:

\begin{equation}
\begin{array}{cc}
   \underbrace{\star\cdots\star}   & =\exp\{\hbar(\sum\limits_{1\leqslant i<j\leqslant m}b_{ij})\}.    \\
      m-times&   
\end{array}
\end{equation}
More precisely, the star product with tensor form can be defined to be

\begin{equation}
(f_{1}(z_{1})\star\cdots\star f_{m}(z_{m}))_{\otimes}=\mathcal{U}
(\exp\{\hbar(\sum\limits_{1\leqslant i<j\leqslant m}b_{ij}),\mathcal{K}\})
(f_{1}(z_{1})\otimes\cdots\otimes f_{m}(z_{m})),
\end{equation}
Where $f_{i}(\cdot)\in C^{\infty}(\mathbb{R}),\,i=1,\cdots,m$.

\begin{remark}
If we consider more general star product with tenser form
$$
(f_{1}(\xi_{1})\otimes\cdots\otimes f_{k}(\xi_{k}))
\star(f_{k+1}(\xi_{k+1})\otimes\cdots\otimes f_{k+l}(\xi_{k+l}))
$$
from viewpoint of Kontsevich graphs,
where $f_{i}\in C^{\infty}(\mathbb{R})$, 
we need to make additional restriction on the graphs
of Bernoulli type. For a graph of Bernoulli type $b_{ij}\in G_{1,m}$
($m=k+l$),
the set of vertices of the second type is divided into
left part and right part. We label the left part by
$\{1,\cdots, k\}$ and $i-$th vertex in left part is assigned to
function $f_{i}(\xi_{i})$.
Similarly, we label right part by $\{k+1,\cdots,k+l\}$,
and $(j+k)-$th vertex in right part is assigned to the
function $f_{k+j}(\xi_{k+j})$. The edges starting at vertex
of the first type in $b_{ij}$ end at $i-$th vertex of the second type 
in left part and $(k+j)-$th vertex of the second type in right part
respectively. Therefore $b_{i,k+j}$ is assigned to
 $\mathcal{K}_{i,k+j}\frac{\partial}{\partial\xi_{i}}\otimes\frac{\partial}{\partial\xi_{k+j}}$.
Then we have

\begin{equation}
\begin{array}{c}
   (f_{1}(\xi_{1})\otimes\cdots\otimes f_{k}(\xi_{k}))
\star(f_{1}(\xi_{k+1})\otimes\cdots\otimes f_{k+l}(\xi_{k+l}))       \\
=\mathcal{U}(\exp\{\hbar(\sum\limits_{1\leqslant i\leqslant k,\,1\leqslant j\leqslant l}b_{i,k+j})\}) 
(f_{1}(\xi_{1})\otimes\cdots\otimes f_{k+l}(\xi_{k+l})).        
\end{array}
\end{equation}
\end{remark}

It is obvious that the star product (5.3) is associative. Let
$I_{1},\cdots,I_{k}$ be a partition of $\{1,\cdots,m\}$ satisfying
$p<q$ if $p\in I_{i},\,q\in I_{j}$ and $i<j$, 
it is easy to check that

$$
f_{I_{1},\otimes}\star\cdots\star f_{I_{k},\otimes}
=(f_{1}(z_{1})\star\cdots\star f_{m}(z_{m}))_{\otimes},
$$
where
$f_{I_{j},\otimes}=(f_{i_{1}}(z_{i_{1}})\star
\cdots\star f_{i_{j}}(z_{i_{j}}))_{\otimes}$
and $I_{j}=\{i_{1},\cdots,i_{j}\}$($1\leqslant j\leqslant k$).
  
The explicit expansion of the star product (5.3) 
is given by the following formula:

\begin{equation}
\begin{array}{c}
(f_{1}(z_{1})\star\cdots\star f_{m}(z_{m}))_{\otimes}   \\
=\sum\limits_{M\in M_{adj}(m,\mathbb{N})}
\frac{\hbar^{degM}}{M!}K_{M}\partial_{\otimes}^{\alpha_{M}}
(f_{1}(z_{1})\otimes\cdots\otimes f_{m}(z_{m})),
\end{array}
\end{equation} 
where $M!=\prod_{1\leqslant i<j\leqslant m}m_{ij}!$,
$K_{M}=\prod_{1\leqslant i<j\leqslant m}K_{ij}^{m_{ij}}$,
$\partial_{\otimes}^{\alpha_{M}}=\partial_{1}^{\alpha_{M,1}}
\otimes\cdots\otimes\partial_{m}^{\alpha_{M,m}}$,
$\alpha_{M}=(\alpha_{M,1},\cdots,\alpha_{M,m})$,
$\alpha_{M,i}=\sum_{j}m_{ij},\,i=1,\cdots, m$.

The formula (5.5) can be regarded as generalized 
Wick expansion. Particularly, if we take $f_{i}(z_{i})=\frac{z_{i}^{n_{i}}}{n_{i}!}$,
($n_{i}\in \mathbb{N},\,i=1,\cdots, m$) we have

\begin{equation}
\begin{array}{c}
     (\frac{z_{1}^{n_{1}}}{n_{1}!}\star\cdots\star \frac{z_{m}^{n_{m}}}{n_{m}!})_{\otimes} \\
   =  \sum\limits_{M\in M_{adj}(m,\mathbb{N})}\hbar^{degM}
    \frac{K_{M}}{M!} 
    \frac{z_{1}^{n_{1}-\alpha_{M,1}}}{(n_{1}-\alpha_{M,1})!}
    \otimes\cdots\otimes \frac{z_{m}^{n_{m}-\alpha_{M,m}}}{(n_{m}-\alpha_{M,m})!}.
\end{array}
\end{equation}

We define the star product in ordinary sense to be
\begin{equation}
f_{1}(z_{1})\star\cdots\star f_{m}(z_{m})
=\textbf{m}\circ (f_{1}(z_{1})\star\cdots\star f_{m}(z_{m}))_{\otimes},
\end{equation}
where $\textbf{m}$ means taking multiplication
of point-wise for functions.
All of previous discussions are still available, but
the tenser will be replaced by point-wise multiplication
of functions.

We can introduce the notation of expectation of star product of
monomials, as what has been done in \cite{8}, which 
will be useful for discussion below.
\begin{definition}
We say a integer sequence $(n_{1},\cdots, n_{m})$ is admissible
if there is an adjacency matrix $M=(m_{ij})_{m\times m}$ such that
\begin{equation}
n_{i}=\sum_{j}m_{ij},i=1,\cdots, m.
\end{equation}
We say such an adjacency matrix $M$ satisfying (5.8) subordinates
the admissible integer sequence as above. We denote it by
$M\prec (n_{1},\cdots, n_{m})$.
\end{definition}

We now define the expectation of star product monomial as following:
\begin{definition}
Let $\frac{z_{1}^{n_{1}}}{n_{1}!}\star\cdots\star \frac{z_{m}^{n_{m}}}{n_{m}!}$ be a star product
monomial, its expectation denoted by 
$<\frac{z_{1}^{n_{1}}}{n_{1}!}\star\cdots\star \frac{z_{m}^{n_{m}}}{n_{m}!}>$ 
is defined to be
\begin{itemize}
  \item When $(n_{1},\cdots,n_{m})$ is an admissible integer sequence,
\begin{equation}
  <\frac{z_{1}^{n_{1}}}{n_{1}!}\star\cdots\star \frac{z_{m}^{n_{m}}}{n_{m}!}>=
 \sum\limits_{M\prec (n_{1},\cdots,n_{m})}\frac{K_{M}}{M!}.
\end{equation}
  \item  
  $$<\frac{z_{1}^{n_{1}}}{n_{1}!}\star\cdots\star \frac{z_{m}^{n_{m}}}{n_{m}!}>=0$$
  for otherwiae.
  \end{itemize}
\end{definition}
About the expectation of star product monomial we have the following
theorem:
\begin{theorem}
An integer sequence $(n_{1},\cdots,n_{m})$ is admissible if and only if
we have
\begin{equation}
z_{1}^{n_{1}}\star\cdots\star z_{m}^{n_{m}}=\hbar^{k}<z_{1}^{n_{1}}\star\cdots\star z_{m}^{n_{m}}>
+terms\, \,with\,\,lower\,\,powe\,r\,than\,\,\hbar^{k},
\end{equation}
where $2k=n_{1}+\cdots+n_{m}$.
\end{theorem}

The proof of theorem 5.1 refers to proposition 4.3 
and Theorem 4.2 in \cite{8}.
With the help of expectation Wick expansion can be expressed in more
classical way,
\begin{equation}
\frac{z_{1}^{n_{1}}}{n_{1}!}\star\cdots\star \frac{z_{m}^{n_{m}}}{n_{m}!}=
\sum\limits_{k\geqslant 0}\frac{\hbar^{k}}{k!}
\sum\limits_{M\prec \alpha,|\alpha|=2k, \alpha_{i}+\beta_{i}=n_{i}}
<\frac{z_{1}^{\alpha_{1}}}{\alpha_{1}!}\star\cdots\star \frac{z_{m}^{\alpha_{m}}}{\alpha_{m}!}>
\frac{z_{1}^{\beta_{1}}}{\beta_{1}!}\cdots \frac{z_{m}^{\beta_{m}}}{\beta_{m}!},
\end{equation}
where $\alpha=(\alpha_{1},\cdots,\alpha_{m})$ expressed by means of notion 
of multiple index.
The integer sequence $(\alpha_{1},\cdots,\alpha_{m})$ in (5.11) are
admissible naturally.

We now turn to the star product at level of fields.
Here we restrict us to consider only the case of 
point-wise multiplication. The case of tenser form is
similar. The star product at level of fields is defined to be
\begin{equation}
f_{1}(\varphi(x_{1}))\star\cdots\star f_{m}(\varphi(x_{m}))=
(f_{1}(z_{1})\star\cdots\star f_{m}(z_{m}))\mid_{z_{i}=\varphi(x_{i})}.
\end{equation}
Where $\varphi(\cdot)$ is real scalar field. Comparing with $T-$product
in quantum field theory we can see that the star product (5.12) is very
similar to $T-$product. For example, the commutativity corresponds
to symmetrical property of $T-$product. The associativity of the star product
corresponds to the factorization of $T-$product. Furthermore, the expectation
of the star product monomials in the case of scalar field can be defined as same 
way as above. We have
\begin{equation}
<\varphi^{n_{1}}(x_{1})\star\cdots\star \varphi^{n_{m}}(x_{m})>=
<z_{1}^{n_{1}}\star\cdots\star z_{m}^{n_{m}}>.
\end{equation}
Moreover, by definition (5.12), we have also Wick expansion as following:
\begin{equation}
\begin{array}{c}
     (\varphi^{n_{1}}(x_{1})/n_{1}!)\star\cdots\star (\varphi^{n_{m}}(x_{m})/n_{m}!) \\
   =  \sum\limits_{k\geqslant 0}\frac{\hbar^{k}}{k!}
   \sum\limits_{M\prec \alpha,|\alpha|=2k, \alpha_{i}+\beta_{i}=n_{i}}
   <\frac{\varphi^{\alpha_{1}}(x_{1})}{\alpha_{1}!}\star\cdots\star \frac{\varphi^{\alpha_{m}}(x_{m})}{\alpha_{m}!}>\\
   (\varphi^{\beta_{1}}(x_{1})/\beta_{1}!)\cdots (\varphi^{\beta_{m}}(x_{m})/\beta_{m}!).
\end{array}
\end{equation}

Similarly, we can define the star product
$$
(f_{1}(\varphi(x_{1}))\cdots f_{k}(\varphi(x_{k})))\star 
(g_{1}(\varphi(y_{1}))\cdots g_{l}(\varphi(y_{l})))
$$
by means of the formula (5.4).

\subsection{Quotient or collapsing of star product}

In this subsection we will discuss the quotient or
collapsing of the star product which is compatible
with the similar notations of adjacency matrices
and Feynman diagrams.

\paragraph{Adjacency matrices and Feynman diagrams:}

The quotient
in the situation of the star product will involve
four objects, as shown in the following diagram:

$$
\begin{array}{ccc}
adjacency\,\,matrx & \longleftrightarrow
& Feynman\,\,amplitude \\
\updownarrow &  &  \updownarrow \\
Bernoulli\,\,graph & \longleftrightarrow &
Feynman\,\,diagram.
\end{array}
$$
Recalling the contents of \cite{8}, there is an
one-one correspondence between Bernoulli graphs and 
Feynman diagrams. Precisely, each vertex of the 
second type (or a vertex of the first type) of 
a Bernoulli graph corresponds to
a vertex (or an internal line) of a Feynman diagram.
Thus, in the latter discussion, a Bernoulli graph will
mean a Feynman diagram, and vice-versa.

As preparation we talk about some notations firstly.
Let $m$ be a positive integer, $\underline{m}=\{1,\cdots,m\}$.
We consider the power
set $\mathcal{P}(\underline{m})$, and
the elements of the power set are labled by
$\{1^{\ast},\cdots, 2^{m,\ast}\}$, such that
each subset of $\underline{m}$ assigns a
number $i^{\ast}$ ($1\leq i\leq 2^{m}$).
Recalling the previous discussions about the quotient,
the elements in $\mathcal{P}(\underline{m})$ indicate
the ideal part arising from the quotient,
thus we call index $i^{\ast}$ ($1\leq i\leq 2^{m}$)
the ideal index. Due to the consideration
of the star product, it is necessary to introduce
the variables and coefficients of Poisson bi-vector
corresponding to ideal indices. For a given $I\subset\underline{m}$
indicated by $i^{\ast}$, let $I$ corresponding to 
a variable $\zeta_{I}$ denoted by $\zeta_{i}$
also. Subsequently, to describe the quotient
of the star product, let $\mathcal{A}$ be an algebra 
over $\mathbb{C}$ (or $\mathbb{R}$) with generators
$\{K_{ij}|1\leq i,j\leq m\}\cup\{K_{i^{\ast}j}|1\leq i\leq 2^{m},
1\leq j\leq m \}\cup\{K_{ij^{\ast}}|1\leq i\leq m,
1\leq j\leq 2^{m}\}$.

Noting that all of discussions in section 3 are available
for $M_{adj}(m,\mathbb{N})$, actually, let 
$M\in M_{adj}(m,\mathbb{N})$,
for $I\subset \underline{m}$,
($|I|=k,1<k<m$), it is obvious that
$$
M_{I}\in M_{adj}(k,\mathbb{N}),\,
M\diagup M_{I}\in M_{adj}(m-k+1,\mathbb{N}).
$$
$M_{adj}(m,\mathbb{N})$ is a monoid under the
addition of the matrices. By the same reason
in section 4.1, we know that, for a given subset
$I\subset \underline{m}$, the quotient

$$
\cdot\diagup I:M_{adj}(m,\mathbb{N})
\longrightarrow M_{adj}(m-k+1,\mathbb{N})
$$
is a homomorphism, and

$$
deg(M\diagup M_{I})=degM-degM_{I},
$$
where $degM=\sum_{1\leq i<j\leq m}m_{ij}$.
Therefore, similar to Hopf algebra $\mathcal{H}_{gl}$,
we can construct a Hopf algebra related to adjacency
matrices, it is enough for us to take $M_{adj}(m,\mathbb{N})$
instead of $gl(m,\mathbb{C})_{0}$ everywhere.
We denote this Hopf algebra by $\mathcal{H}_{adj}$.

Now we turn to discuss the Feynman diagrams. 
Here we restrict
us to discuss the subgraphs of Feynman diagrams.
A subdiagram of Feynman diagram is subset of
vertices and lines in Feynman diagram, the lines in subdiagram
join the vertices in subdiagram. A subgraph is
a subdiagram but the line joining tow vertices
in this subdiagram should belong to the subdiagram.
Thus a subgraph determined by vertices solely.
We discuss the
problems starting at the Bernoulli graphs.
Let $M\in M_{adj}(m,\mathbb{N})$, then 
$b_{M}\in B_{k,m}$, where $k=degM$.
We know that $b_{M}$ can be regarded as a Feynman
diagram with $m$ vertices and $k$ internal lines. 
A subgraph can be identified
with a subset $I$ in $\underline{m}=\{1,\cdots,m\}$
($|I|=k,\,2\leq k< m$), therefore,
$b_{M_{I}}$ just be this subgraph. This subgraph
gives a factorisation of $b_{M}$:
$$
b_{M}=b_{M_{I}}(\prod\limits_{i\in I,\,j\in I^{c}}b_{ij}^{m_{ij}})b_{M_{I^{c}}},
$$
where $I^{c}=\underline{m}\setminus I$.
We make quotient of $b_{M}$ by $b_{M_{I}}$ in the following way:
\begin{itemize}
  \item Dropping the factor $b_{M_{I}}$ and collapsing
  the subset $I$ to an "ideal vertex" of the second type, 
  denoted by $\ast$ (which is also called
  a "generalised point" by Bogoliubov), we get the quotient
  \begin{equation}
  b_{M}\diagup b_{M_{I}}=\prod\limits_{j\in I^{c}}
  b_{\ast,j}^{m_{\ast,j}}b_{M_{I^{c}}},
  \end{equation}
  where $m_{\ast,j}=\sum_{i\in I}m_{ij}$. 
  \item The formula (5.15) shows that $b_{M}\diagup b_{M_{I}}$
  is nothing else but $b_{M\diagup M_{I}}$. 
  Thus, if $b_{M}\in B_{l_{1},m}$
  and $b_{M_{I}}\in B_{l_{2},k}$, then 
  $b_{M}\diagup b_{M_{I}}\in B_{l_{1}-l_{2},m-k+1}$.
  Recalling $B_{m}=(\bigcup_{l}B_{l,m})\cup\{\emptyset\}$
  is a monoid, then, for a given subset $I\subset\underline{m}$
  as mentioned above, the quotient
  
  $$
  \cdot\diagup I:B_{m}\longrightarrow B_{m-k+1}
  $$
  is a homomorphism.
\end{itemize}

Generally, for a partition $(I_{i})$ in $\underline{m}$ ($|I_{i}|>1$) ,
we have 
$$
b_{M}\diagup b_{M_{(I_{i})}}=b_{M\diagup (I_{i})},
$$
where subset $I_{i}$ collapses to $i-$th "ideal vertex"
corresponding to $i-$th "ideal index" of rows or columns
of $M\diagup M_{(I_{i})}$, and the "ideal index" 
of rows or columns of $M\diagup M_{I_{i}}$
corresponds to the index of subset $I_{i}$.

\begin{remark}
From the previous discussion we know that
$\mathcal{H}_{adj}$ just be the Hopf algebra of
Feynman diagrams, denoted by $\mathcal{H}_{Fey}$.
The multiplication in $\mathcal{H}_{Fey}$
is disjoint union of two diagrams and addition is 
formal one. The coproduct for a Feynman diagram $\Gamma$ 
is defined to be
\begin{equation}
\triangle \Gamma=\Gamma\otimes\emptyset+\emptyset\otimes \Gamma
+\sum\limits_{\gamma\subset \Gamma}\gamma\otimes (\Gamma\diagup\gamma).
\end{equation} 
The sum on right side of (5.16) is over all non-trivial subgraphs in $\Gamma$,
here we do not make restriction demanding
the subgraphs are sub-divergent.
\end{remark}

The quotient of Feynman amplitudes should satisfy

$$
\mathcal{U}(b_{M\diagup M_{I}},\mathcal{K})
=\mathcal{K}_{M}\diagup\mathcal{K}_{M_{I}}.
$$
Noting previous
discussion about the quotient of Feynman diagrams,
we have

\begin{equation}
\mathcal{K}_{M}\diagup\mathcal{K}_{M_{I}}
=\mathcal{K}_{M\diagup M_{I}}
=\prod\limits_{j}\mathcal{K}_{\ast,j}^{m_{\ast,j}}
\mathcal{K}_{M_{I^{c}}},
\end{equation}
where $\mathcal{K}_{ij}=K_{ij}\partial_{i}\partial_{j}$,
$\mathcal{K}_{\ast,j}=K_{\ast,j}
\partial_{\zeta_{\ast}}\partial_{j}$.

\paragraph{Quotient or collapsing concerning star product:}

Now we will generalise the notations of quotient and collapsing
to the case of star product. For convenience we discuss the 
case of star product at level of functions
Same as subsection 5.2 we assign $i-$th
vertex to a smooth function $f_{i}(\cdot)$ and a variable $z_{i}$
($1\leqslant i\leqslant m$).
In addition, for a subset $I$ of $\underline{m}$, 
we assign $I$ to variable $\zeta_{I}$,
where $\zeta_{I}$ can be also denoted by $\zeta_{i^{\ast}}$ 
if $I$ is indicated by index $i^{\ast}$.

A subset $I=\{i_{1},\cdots,i_{k}\}\subset\underline{m}$ 
assigns two functions which are:

$$
f_{I,\star}(\textbf{z}_{I})=f_{i_{1}}(z_{i_{1}})
\star\cdots\star f_{i_{k}}(z_{i_{k}}),
$$
and

$$
f_{I}(\zeta_{I})=f_{i_{1}}(\zeta_{I})
\cdots f_{i_{k}}(\zeta_{I}),
$$
where $\textbf{z}_{I}=(z_{i_{1}},\cdots,z_{i_{k}})$.
Now we explain the difference between 
$f_{I,\star}(\textbf{z}_{I})$ and
$f_{I}(\zeta_{I})$ more clearly. 
When $I$ is regarded as a subset of $\underline{m}$,
equivalently, $I$ is regarded as a subgraph of some
Feynman diagram, it assigns to $f_{I,\star}(z)$. 
When $I$ is regarded
as an element in $\mathcal{P}(\underline{m})$, $I$ plays
the role of "ideal vertex" of a Feynman diagram arising from
quotient, i.e. the subgraph $I$ collapses to a "ideal vertex"
$\{I\}$ in quotient, thus, $I$ assigns to function 
$f_{I}(\zeta_{I})$. More general, for an element 
$\{U\}\cup\{I_{1},\cdots,I_{l}\}\in\Xi_{\underline{m}}$,
we assign it to the following star product

\begin{equation}
f_{\{U\}\cup\{I_{i}\},\star}(z,\zeta)=f_{I_{1}}(\zeta_{I_{1}})\star\cdots
\star f_{I_{l}}(\zeta_{I_{l}})\star f_{U,\star}(\textbf{z}_{U}).
\end{equation}
In this situation, the star product is defined as following

$$
\begin{array}{c}
g_{1}(\zeta_{1})\star\cdots\star g_{l}(\zeta_{l})\star
f_{1}(z_{1})\star\cdots\star f_{m}(z_{m})  \\
=\exp\{\hbar(\sum\limits_{i<j}
\mathcal{K}_{i^{\ast}j^{\ast}}
+\sum\limits_{i,j}\mathcal{K}_{i^{\ast}j}+
\sum\limits_{i<j}\mathcal{K}_{ij})\}
(g_{1}(\zeta_{1})\cdots f_{m}(z_{m})),
\end{array}
$$
where $\mathcal{K}_{i^{\ast}j^{\ast}}=
K_{i^{\ast}j^{\ast}}\partial_{\zeta_{i}}\partial_{\zeta_{j}}$,
$\mathcal{K}_{i^{\ast}j}=K_{i^{\ast}j}\partial_{\zeta_{i}}
\partial_{z_{j}}$, $\mathcal{K}_{ij}=K_{ij}
\partial_{z_{i}}\partial_{z_{j}}$.

Now we consider the star product

$$
f_{1}(z_{1})\star\cdots\star f_{m}(z_{m})
=\exp\{\hbar\sum\limits_{1\leqslant i<j\leqslant m}
\mathcal{K}_{ij}\}(f_{1}(z_{1})\cdots f_{m}(z_{m})).
$$
let $I\subset \underline{m}$, $I=\{i_{1},\cdots,i_{k}\}$ 
($1\leqslant i_{1}<\cdots<i_{k}\leqslant m,\,1<k<m$), 
similar to the case of Feynman diagrams we have a
factorisation:

$$
f_{\underline{m},\star}(\textbf{z})=f_{I,\star}
(\textbf{z}_{I})\star f_{I^{c},\star}(\textbf{z}_{I^{c}}).
$$
Without loss of generality, we assume $I=\{1,\cdots,k\}$,
and subset $I$ is labeled by $1^{\ast}$,
by the definition of star product, the formula (5.5), 
and the associativity of the star product, we have

$$
f_{I,\star}(\textbf{z}_{I})\star f_{I^{c},\star}(\textbf{z}_{I^{c}})
=\exp\{\hbar\sum\limits_{1\leqslant i\leqslant k,\,k+1\leqslant j\leqslant m}
K_{ij}\partial_{i}\partial_{j}\}
(f_{I,\star}(\textbf{z}_{I})f_{I^{c},\star}(\textbf{z}_{I^{c}})).
$$
The procedure of collapsing is shown as follows:
\begin{itemize}
   \item The indices in $I$ collapse to a "ideal index" $1^{\ast}$.
   \item $\mathcal{K}_{ij}$ collapse to $\mathcal{K}_{1^{\ast}j}$, i.e.
   $K_{ij}$ collapse to
   $K_{1^{\ast}j}$ ($1\leqslant i\leqslant k,\,k+1\leqslant j\leqslant m$),
   and partial derivatives $\partial_{1},\cdots,\partial_{k}$ 
   collapse to $\partial_{\zeta_{1}}$.
   \item We insert a factor $f_{I}(\zeta_{1})$, corresponding to
   "ideal vertex" $1^{\ast}$,
   into the expression of above factorisation.
\end{itemize}   
Thus we have:

$$
\begin{array}{c}
\exp\{\hbar\sum\limits_{1\leqslant i\leqslant k,\,k+1\leqslant j\leqslant m}
\mathcal{K}_{ij}\partial_{i}\partial_{j}\}
(f_{I,\star}(\textbf{z}_{I})f_{I^{c},\star}(\textbf{z}_{I^{c}}))  \\
\downarrow\,\,collapsing         
\end{array}
$$

\begin{equation}
f_{I,\star}(\textbf{z}_{I})
\exp\{\hbar\sum\limits_{j\in I^{c}}\mathcal{K}_{1^{\ast}j}\partial_{j}
\sum\limits_{i\in I}\partial_{i}\}
(f_{1}(z_{1})\cdots f_{k}(z_{k})f_{I^{c},\star}(\textbf{z}_{I^{c}})))|_{z_{1}=\cdots=z_{k}=\zeta}.
\end{equation}
The right factor in
expression (5.19) is called quotient of $f_{\underline{m},\star}(\textbf{z})$
by $f_{I,\star}(\textbf{z}_{I})$ denoted by $f_{\underline{m},\star}(\textbf{z})\diagup I$,
i.e. by dropping the factor $f_{I,\star}(\textbf{z}_{I})$ in (5.19) we reach
the definition of the quotient in the situation of star product.

\begin{equation}
f_{\underline{m},\star}(\textbf{z})\diagup I=
\exp\{\hbar\sum\limits_{j\in I^{c}}\mathcal{K}_{1^{\ast}j}\partial_{j}
\sum\limits_{i\in I}\partial_{i}\}
(f_{1}(z_{1})\cdots f_{k}(z_{k})f_{I^{c},\star}(\textbf{z}_{I^{c}}))
|_{z_{1}=\cdots=z_{k}=\zeta}.
\end{equation}

\begin{proposition}
	
	\begin{equation}
	f_{\underline{m},\star}(\textbf{z})\diagup I 
	=f_{I}(\zeta_{1})\star f_{I^{c},\star}(\textbf{z}_{I^{c}}).
	\end{equation}
\end{proposition}

To prove proposition 5.1 we need the floowing
obvious formula:

\begin{lemma}
$$
(\sum\limits_{i\in I}\partial_{i})^{l}(f_{1}(z_{1})\cdots 
f_{k}(z_{k}))|_{z_{1}=\cdots=z_{k}=\zeta_{1}}
=(f_{1}(\zeta_{1})\cdots f_{k}(\zeta_{1}))^{(l)}=(f_{I}(\zeta_{1}))^{(l)}.
$$
\end{lemma}

\begin{proof}
$\mathbf{(Proof\,\,of\,\,proposition\,\,5.1)}$
Observing the formula (5.20), we have

$$
\begin{array}{c}
\exp\{\hbar\sum\limits_{j\in I^{c}}\mathcal{K}_{1^{\ast}j}\partial_{j}
\sum\limits_{i\in I}\partial_{i}\}
(f_{1}(z_{1})\cdots f_{k}(z_{k})f_{I^{c},\star}(\textbf{z}_{I^{c}})) \\
=\sum\limits_{k\geq 0}\frac{\hbar^{k}}{k!}
(\sum\limits_{j\in I^{c}}\mathcal{K}_{1^{\ast}j}\partial_{j}
\sum\limits_{i\in I}\partial_{i})^{k}
(f_{1}(z_{1})\cdots f_{k}(z_{k})f_{I^{c},\star}(\textbf{z}_{I^{c}})) \\
=\sum\limits_{k\geq 0}\frac{\hbar^{k}}{k!}
(\sum\limits_{i\in I}\partial_{i})^{k}
(f_{1}(z_{1})\cdots f_{k}(z_{k}))
(\sum\limits_{j\in I^{c}}\mathcal{K}_{1^{\ast}j}\partial_{j})^{k}
f_{I^{c},\star}(\textbf{z}_{I^{c}}).
\end{array}
$$
By lemma 5.1 we know that

$$
(\sum\limits_{i\in I}\partial_{i})^{k}
(f_{1}(z_{1})\cdots f_{k}(z_{k}))|_{z_{1}=\cdots=z_{k}=\zeta_{1}}
=(f_{I}(\zeta_{1}))^{(k)}.
$$
Finally, we get

$$
\begin{array}{c}
\exp\{\hbar\sum\limits_{j\in I^{c}}\mathcal{K}_{1^{\ast}j}\partial_{j}
\sum\limits_{i\in I}\partial_{i}\}
(f_{1}(z_{1})\cdots f_{k}(z_{k})f_{I^{c},\star}(\textbf{z}_{I^{c}}))
|_{z_{1}=\cdots=z_{k}=\zeta} \\
=\sum\limits_{k\geq 0}\frac{\hbar^{k}}{k!}
\partial_{\zeta_{1}}^{k}f_{I}(\zeta_{1})
(\sum\limits_{j\in I^{c}}\mathcal{K}_{1^{\ast}j}\partial_{j})^{k}
f_{I^{c},\star}(\textbf{z}_{I^{c}}) \\
=\sum\limits_{k\geq 0}\frac{\hbar^{k}}{k!}
(\sum\limits_{j\in I^{c}}\mathcal{K}_{1^{\ast}j}
\partial_{\zeta_{1}}\partial_{j})^{k}
(f_{I}(\zeta_{1})f_{I^{c},\star}(\textbf{z}_{I^{c}}))
=f_{I}(\zeta_{1})\star f_{I^{c},\star}(\textbf{z}_{I^{c}}).
\end{array}
$$

\end{proof}

We hope to compare the quotient of star product
with one of adjacency matrices and Feynman diagrams
to show that they are compatible. We now take
one term from the expression of the star product

$$
f_{\underline{m},\star}(\textbf{z})=
\sum\limits_{M\in M_{adj}(m,\mathbb{N})}
\frac{\hbar^{degM}}{M!}\mathcal{K}_{M}
(f_{1}(z_{1})\cdots f_{m}(z_{m})),
$$
that is

$$
\frac{\hbar^{degM}}{M!}\mathcal{K}_{M}
(f_{1}(z_{1})\cdots f_{m}(z_{m})).
$$
We are interested in how the factorization
is shown in this situation. At level of the
adjacency matrix, the factorization should be
expressed as $M=M_{1}+\iota_{I}M_{I}+\iota_{I^{c}}M_{I^{c}}$,
$\iota_{I}:M_{adj}(k,\mathbb{N})\rightarrow M_{adj}(m,\mathbb{N})$ and 
$\iota_{I^{c}}:M_{adj}(m-k,\mathbb{N})\rightarrow M_{adj}(m,\mathbb{N})$
are defined as in section 4.1.
It is obvious that $M!=M_{1}!M_{I}!M_{I^{c}}!$,
$degM=degM_{I}+deg(M-\iota_{I}M_{I})$, and
$\mathcal{K}_{M}=\mathcal{K}_{M_{1}}
\mathcal{K}_{M_{I}}\mathcal{K}_{M_{I^{c}}}$.
Now we insert a factor $\prod_{i\in I}f_{i}(z_{i})$,
then we have,

$$
\frac{\hbar^{degM_{I}}}{M_{I}!}
\mathcal{K}_{M_{I}}(\prod_{i\in I}f_{i}(z_{i}))
\frac{\hbar^{degM_{1}}}{M_{1}!}
\mathcal{K}_{M_{1}}(\prod_{i\in I}f_{i}(z_{i}))
\frac{\hbar^{degM_{I^{c}}}}{M_{I^{c}}!}
\mathcal{K}_{M_{I^{c}}}(\prod_{i\in I^{c}}f_{i}(z_{i})).
$$
The procedure of collapsing takes place in
the middle factor

$$
\begin{array}{ccc}
 & collaping & \\
\frac{\hbar^{degM_{1}}}{M_{1}!}
\mathcal{K}_{M_{1}}(\prod_{i\in I}f_{i}(z_{i})) &
\longrightarrow &
\frac{\hbar^{degM^{\ast}}}{M^{\ast}!}
\mathcal{K}_{M^{\ast}}f_{I}(\zeta_{I}),
\end{array}
$$
where $M^{\ast}$ satisfies 
$M^{\ast}\in M_{adj}(m-k+1,\mathbb{N}),\,
M^{\ast}+\iota_{I^{c}}M_{I^{c}}=M\diagup M_{I}$,
and $\iota_{I^{c}}:M_{adj}(m-k,\mathbb{N})
\rightarrow M_{adj}(m-k+1,\mathbb{N})$.

Dropping the factor $\frac{\hbar^{degM_{I}}}{M_{I}!}
\mathcal{K}_{M_{I}}f_{I}(\textbf{z}_{I})$,
we get

$$
\frac{\hbar^{deg(M\diagup M_{I})}}{(M\diagup M_{I})!}
\mathcal{K}_{M\diagup M_{I}}(f_{I}(\zeta_{I})\prod_{i\in I^{c}}f_{i}(z_{i})).
$$
In summary, we now reach:

\begin{proposition}
\begin{equation}
f_{\underline{m},\star}(\textbf{z})\diagup I=
\sum\limits_{M\in M_{adj}(m,\mathbb{N})}
\frac{\hbar^{deg(M\diagup M_{I})}}{(M\diagup M_{I})!}
\mathcal{K}_{M\diagup M_{I}}
(f_{I}(\zeta_{I})\prod_{i\in I^{c}}f_{i}(z_{i})).
\end{equation}
\end{proposition}

Previous discussion about quotient and collapsing of
star product can be generalised to the case of
partitions. Let $(I_{1},\cdots,I_{l})$ be a partition in $\underline{m}$,
$I=\bigcup_{i=1}^{l}I_{i}$, $I^{c}=\underline{m}\setminus I$,
$|I|=k\,(1<k<m)$, then we have
$$
\underline{m}\diagup(I_{i})
=I^{c}\cup\{I_{i}\}.
$$
Without loss of generality, we assume
the subset $I_{i}$ is labeled by $i^{\ast}$
($i=1,\cdots,l$), which means we identify
$I^{c}\cup\{I_{i}\}$ with $I^{c}\cup\{1^{\ast},\cdots,l^{\ast}\}$.
It is easy to check that
$$
\begin{array}{c}
(f_{\underline{m},\star}(\textbf{z})\diagup I_{1})\diagup I_{2} \\
=f_{I_{1}}(\zeta_{1})\star
(f_{\underline{m}\setminus I_{1},\star}(z)\diagup I_{2})  \\
=f_{I_{1}}(\zeta_{1})\star f_{I_{2}}(\zeta_{2})\star
f_{\underline{m}\setminus (I_{1}\cup I_{2}),\star}
(\textbf{z}_{(I_{1}\cup I_{2})^{c}}).
\end{array}
$$
We denote the quotient of $f_{\underline{m},\star}(\textbf{z})$ by
$(f_{I_{1},\star}(\textbf{z}_{I_{1}}),\cdots,
f_{I_{l},\star}(\textbf{z}_{I_{l}}))$ by 
$f_{\underline{m},\star}(\textbf{z})\diagup (I_{i})$ simply, inductively,
we have formula similar to (5.21) as following:
\begin{equation}
      f_{\underline{m},\star}(\textbf{z})\diagup (I_{i})    
      =f_{I_{1}}(\zeta_{1})\star\cdots\star
     f_{I_{l}}(\zeta_{l})\star f_{I^{c},\star}(\textbf{z}_{I^{c}}),  
\end{equation}
and

\begin{equation}
f_{\underline{m},\star}(\textbf{z})\diagup (I_{i})=
\sum\limits_{M\in M_{adj}(m,\mathbb{N})}
\frac{\hbar^{deg(M\diagup(I_{i}))}}{(M\diagup(I_{i}))!}
\mathcal{K}_{M\diagup(I_{i})}(f_{I_{1}}(\zeta_{1})\cdots
f_{I_{l}}(\zeta_{l})\prod\limits_{i\in I^{c}}f_{i}(z_{i})).
\end{equation}

Let $(J_{1},\cdots,J_{k})$ be a partition in
$\underline{m}\diagup(I_{i})$,
we consider $(f_{\underline{m},\star}(\mathbf{z})\diagup(I_{i}))
\diagup(J_{j})$. If we take
$$
(K_{\lambda})=(I_{i})_{I_{i}
\notin\mathcal{R}(\{J_{j}\})}
\cup(\mathcal{R}(J_{j})),
$$
we can prove

\begin{proposition}

\begin{equation}
(f_{\underline{m},\star}(\mathbf{z})\diagup(I_{i}))
\diagup(J_{j})=f_{\underline{m},\star}(\mathbf{z})\diagup
(K_{\lambda}).
\end{equation}
\end{proposition}

The proof of (5.25) is similar to the situations
in previous sections.

Combining (5.18) and (5.23) we know that
$$
f_{\{U\}\cup\{I_{i}\},\star}(\mathbf{z}_{U},
\zeta_{I_{1}},\cdots,\zeta_{I_{l}})=
f_{V,\star}(\mathbf{z}_{V})\diagup (I_{i}),
$$
where $V=U\cup\mathcal{R}(\{I_{i}\})$,
$l=|\{I_{i}\}|$. Moreover, we assign a sequence 
$$
\{D_{1}\cup E_{1},\cdots,D_{k}\cup E_{k}\}
\in\mathcal{P}_{dis}(\Xi_{\underline{m}})
$$
to 
$$
f_{\{D_{j}\cup E_{j}\},\star}
=\prod\limits_{j}f_{D_{j}\cup E_{j},\star}
(\mathbf{z}_{D_{j}},\zeta_{E_{j}}),
$$ 
where the multiplication is point-wise one of functions
and each factor $f_{D_{j}\cup E_{j},\star}
(\mathbf{z}_{D_{j}},\zeta_{E_{j}})$ is given by
(5.18), for example, if $E_{j}=\{I_{1},\cdots,I_{k_{j}}\}$,
then

$$
f_{D_{j}\cup E_{j},\star}
(\mathbf{z}_{D_{j}},\zeta_{E_{j}})=
f_{I_{1}}(\zeta_{I_{1}})\star\cdots\star 
f_{I_{k_{j}}}(\zeta_{I_{k_{j}}})
\star f_{D_{j}}(\mathbf{z}_{D_{j}}).
$$
We assume $(D_{j}\cup E_{j})\subset U\cup\{I_{i}\}$
and discuss the quotient 
$$
\begin{array}{c}
f_{\{U\}\cup\{I_{i}\}}(z,\zeta)\diagup(D_{j}\cup E_{j}) \\
=(\cdots(f_{\{U\}\cup\{I_{i}\}}(z,\zeta)\diagup D_{1}\cup E_{1})
\diagup\cdots)\diagup D_{k}\cup E_{k}.
\end{array}
$$
Let $\{I_{i^{\prime}}\}=\{I_{i}\}\setminus\mathcal{R}(\{E_{j}\})$,
then we have
$$
\begin{array}{c}
f_{\{U\}\cup\{I_{i}\}}(z,\zeta)\diagup(D_{j}\cup E_{j})          \\
=f_{U\setminus\mathcal{R}(\{D_{j}\}),\star}(z)\star
f_{\{I_{i^{\prime}}\}}(\zeta)\star 
f_{\{D_{j}\cup\mathcal{R}(E_{j})\},\star}(\zeta)  \\
=f_{V,\star}(z)\diagup(M_{\mu}),       
\end{array}
$$
where $(M_{\mu})=(I_{i^{\prime}})\cup
(D_{j}\cup\mathcal{R}(E_{j}))$. Particularly, if
$\{I_{i}\}=\mathcal{R}(\{E_{j}\})$, then
$(M_{\mu})=(D_{j}\cup\mathcal{R}(E_{j}))$
and
$$
\begin{array}{c}
f_{\{U\}\cup\{I_{i}\}}(z,\zeta)\diagup(D_{j}\cup E_{j})          \\
=f_{U\setminus\mathcal{R}(\{D_{j}\}),\star}(z)\star 
f_{\{D_{j}\cup\mathcal{R}(E_{j})\},\star}(\zeta) \\  
=f_{U\cup\mathcal{R}(\{I_{i}\}),\star}(z)
\diagup(D_{j}\cup\mathcal{R}(E_{j})).    
\end{array}
$$

Let $\{K_{\lambda}\cup L_{\lambda}\}\in
\mathcal{P}_{dis}(\Xi_{\underline{m}})$ such that
$(K_{\lambda}\cup L_{\lambda})\subset(D_{j}\cup E_{j})$,
then we have decomposition
$$
(K_{\lambda}\cup L_{\lambda})=
\bigcup\limits_{i}(K_{\lambda_{ij}}\cup L_{\lambda_{ij}}),
(K_{\lambda_{ij}}\cup L_{\lambda_{ij}})\subset D_{i}\cup E_{i}.
$$
It is natural for us to define the following quotient:
$$
f_{\{D_{i}\cup E_{i}\},\star}(z,\zeta)\diagup(K_{\lambda}\cup L_{\lambda})
=\prod\limits_{i}f_{D_{i}\cup E_{i},\star}(z,\zeta)\diagup
(K_{\lambda_{ij}}\cup L_{\lambda_{ij}}).
$$
If $\mathcal{R}(\{E_{j}\})=\mathcal{R}(\{L_{\lambda}\})$,
from above discussion we know that
$$
f_{\{D_{i}\cup E_{i}\},\star}(z,\zeta)\diagup(K_{\lambda}\cup L_{\lambda})
=\prod\limits_{i}f_{D_{i}\cup\mathcal{R}(E_{i}),\star}(z)\diagup
(K_{\lambda_{ij}}\cup\mathcal{R}(L_{\lambda_{ij}})).
$$

\subsection{Hopf algebra}

The discussions in subsection 5.3 indicate that
there is a map 
$F_{\star}:\mathcal{P}_{dis}(\Xi_{\underline{m}})
\to\mathbf{C}_{\mathcal{A},\hbar}^{\infty}$,
thus, a linear extension of $F_{\star}$ 
denoted by $F_{\star}$ also,
$F_{\star}:\mathbf{Span}_{\mathbb{K}}
(\mathcal{P}_{dis}(\Xi_{\underline{m}}))
\to\mathbf{Span}_{\mathbb{K}}
(\mathbf{C}_{\mathcal{A},\hbar}^{\infty})$,
where $\mathbb{K}=\mathbb{R}$ or $\mathbb{C}$
and

$$
\mathbf{C}_{\mathcal{A},\hbar}^{\infty}=
\{\sum\limits_{k\geq 0}\hbar^{k}c_{k}
\partial^{\alpha_{k}}f_{I_{k}}(\mathbf{z}_{I_{k}})
\partial^{\beta_{k}}f_{J_{k}}(\zeta_{J_{k}})|
c_{k}\in\mathcal{A},\,I_{k}\subset\underline{m},\,
J_{k}\in\mathcal{P}_{dis}^{2}(\underline{m})\}.
$$
Above facts suggests us to construct
the coproduct in the situation of the star 
product based on the discussions in section 3.

\begin{definition}
\begin{itemize}
\item
Let $\{U\}\cup\{I_{i}\}\in\Xi_{\underline{m}}$,
$\{D_{j}\cup E_{j}\}\in\mathcal{P}_{dis}(\Xi_{\underline{m}})$,
$(D_{j}\cup E_{j})\subset U\cup\{I_{i}\}$, 
$\mathcal{R}(\{E_{j}\})=\{I_{i}\}$, we have

\begin{equation}
\triangle_{(D_{j}\cup E_{j})}f_{\{U\}\cup\{I_{i}\},\star}(z,\zeta)
=f_{\{D_{j}\cup E_{j}\},\star}(z,\zeta)\otimes f_{V,\star}(z)\diagup(M_{j}),
\end{equation}
where $V=U\cup\mathcal{R}(\{I_{i}\})$ and
$$
\{M_{j}\}=(id\times\mathcal{R}_{1})
(U\cup\{I_{i}\}\diagup(D_{j}\cup E_{j}))
=\{D_{j}\cup\mathcal{R}(E_{j})\}.
$$
\item Let $\{K_{\lambda}\cup L_{\lambda}\},\,\{D_{j}\cup E_{j}\}
\in\mathcal{P}_{dis}(\Xi_{\underline{m}})$,
$(K_{\lambda}\cup L_{\lambda})\subset(D_{j}\cup E_{j})$,
$\mathcal{R}(\{L_{\lambda}\})=\mathcal{R}(\{E_{j}\})$,
we have

\begin{equation}
\triangle_{(K_{\lambda}\cup L_{\lambda})}
f_{\{D_{j}\cup E_{j}\},\star}=f_{\{K_{\lambda}\cup L_{\lambda}\},\star}
\otimes\prod\limits_{i}f_{D_{i}\cup E_{i},\star}\diagup
(K_{\lambda_{ij}}\cup L_{\lambda_{ij}}),
\end{equation}
where 
$$
(K_{\lambda}\cup L_{\lambda})=\bigcup\limits_{i}
(K_{\lambda_{ij}}\cup L_{\lambda_{ij}}),\,
(K_{\lambda_{ij}}\cup L_{\lambda_{ij}})
\subset D_{i}\cup E_{i}.
$$
\end{itemize}
\end{definition}

\begin{definition}
\begin{itemize}
\item Let $\{U\}\cup\{I_{i}\}\in\Xi_{\underline{m}}$,
we define

\begin{equation}
\begin{array}{c}
\triangle f_{U\cup\{I_{i}\},\star}(z,\zeta)=
f_{U\cup\{I_{i}\},\star}(z,\zeta)\otimes 1+
1\otimes f_{U\cup\{I_{i}\},\star}(z,\zeta) \\
+\sum\limits_{(D_{j}\cup E_{j})\subset U\cup\{I_{i}\},\,
\mathcal{R}(\{E_{j}\})=\{I_{i}\}}\triangle_{(D_{j}\cup E_{j})}
f_{U\cup\{I_{i}\},\star}(z,\zeta).
\end{array}
\end{equation}
\item Let $\{D_{j}\cup E_{j}\}\in
\mathcal{P}_{dis}(\Xi_{\underline{m}})$,
we define

\begin{equation}
\begin{array}{c}
\triangle f_{\{D_{j}\cup E_{j}\},\star}          
=f_{\{D_{j}\cup E_{j}\},\star}\otimes 1
+1\otimes f_{\{D_{j}\cup E_{j}\},\star} \\
+\sum\limits_{(K_{\lambda}\cup L_{\lambda})\subset(D_{j}\cup E_{j}),
\,\mathcal{R}(\{L_{\lambda}\})=\mathcal{R}(\{E_{j}\})}
\triangle_{(K_{\lambda}\cup L_{\lambda})} 
f_{\{D_{j}\cup E_{j}\},\star}.       
\end{array}
\end{equation}
\end{itemize}
\end{definition}

By the same way as the contents in
section 3, we can prove the coproduct
in defintion 5.8 is well defined.
Hence, $\mathbf{Span}_{\mathbb{K}}
(\mathbf{C}_{\mathcal{A},\hbar}^{\infty})$
is a coalgebra. Furthermore, the reduced
coproduct is conilpotent, therefore,
$T(\mathbf{Span}_{\mathbb{K}}
(\mathbf{C}_{\mathcal{A},\hbar}^{\infty}))$
and $S(\mathbf{Span}_{\mathbb{K}}
(\mathbf{C}_{\mathcal{A},\hbar}^{\infty}))$
are two Hopf algebras. On the other hand, we know that
the elments in $S(\mathbf{Span}_{\mathbb{K}}
(\mathbf{C}_{\mathcal{A},\hbar}^{\infty}))$
can be identified with the polynomials of the formal power
series in $\mathbf{C}_{\mathcal{A},\hbar}^{\infty}$
under the multiplication of the formal power
series, thus we have

$$
S(\mathbf{C}_{\mathcal{A},\hbar}^{\infty})
=\mathcal{A}(\mathbf{C}_{\mathcal{A},\hbar}^{\infty}),
$$
where $\mathcal{A}(\mathbf{C}_{\mathcal{A},\hbar}^{\infty})$
is the algebra generated by
$\mathbf{C}_{\mathcal{A},\hbar}^{\infty}$
with the multiplication of the formal power
series.

\end{document}